\documentclass[11pt]{article}
\usepackage{thanks}
\usepackage{theapa}

\normalbaselines

\usepackage{amsmath, amssymb}
\usepackage{paralist}
\usepackage{sgame}
\usepackage{url}

\newcommand{\bfe}[1]{\emph{\bfseries#1}}
\newcommand{\myra}{\mbox{$\:\rightarrow\:$}}
\newcommand{\af}{\ensuremath{\textit{AF}}}
\renewcommand{\vec}[1]{\mathbf{#1}}
\newcommand{\discr}{\ensuremath{\delta}}

\usepackage{amsthm}

\newtheorem{theorem}{Theorem}
\newtheorem{proposition}{Proposition}
\newtheorem{lemma}{Lemma}
\newtheorem{claim}{Claim}

\theoremstyle{definition}
\newtheorem{example}{Example}
\newtheorem{remark}{Remark}

\usepackage{pstricks,pst-node}
\newcommand{\ownnodeL}[6]{
  \pscircle[linewidth=1pt,fillstyle=solid,fillcolor=#2](#3){#1}
  \uput{.3}[#6](#3){#5}
  \pnode(#3){#4}
}

\newcommand{\gs}[1]{#1}

\begin{document}
\title{Selfishness Level of Strategic Games
\thanks{A preliminary version of this paper appeared in the \emph{Proceedings of the 5th International Symposium on Algorithmic Game Theory} \cite{AS12a}.}
}

\author{%
Krzysztof R. Apt 
\thanks{%
 Centre for Mathematics and Computer Science (CWI)
 and University of Amsterdam, The Netherlands}
\and
Guido Sch\"{a}fer
\thanks{%
 Centre for Mathematics and Computer Science (CWI)
 and VU University Amsterdam, The Netherlands}
}

\date{}
\maketitle

\begin{abstract}
  We introduce a new measure of the discrepancy in strategic games
  between the social welfare in a Nash equilibrium and in a social
  optimum, that we call \bfe{selfishness level}. It is the smallest
  fraction of the social welfare that needs to be offered to each
  player to achieve that a social optimum is realized in a pure Nash
  equilibrium. The selfishness level is unrelated to the price of
  stability and the price of anarchy and is invariant under positive linear transformations of the payoff
  functions. Also, it naturally applies to other solution concepts and
  other forms of games.
  
  We study the selfishness level of several well-known strategic
  games.  This allows us to quantify the implicit tension within a
  game between players' individual interests and the impact of their
  decisions on the society as a whole. Our analyses reveal that the
  selfishness level often provides a deeper understanding of the 
  characteristics of the underlying game that influence the players'
  willingness to cooperate.
  
  In particular, the selfishness level of
  finite ordinal potential games is finite, while that of weakly
  acyclic games can be infinite.  We derive explicit bounds on the
  selfishness level of fair cost sharing games and linear congestion
  games, which depend on specific parameters of the underlying game
  but are independent of the number of players.  Further, we show that
  the selfishness level of the $n$-players Prisoner's Dilemma is
  \gs{$c/(b(n-1)-c)$, where $b$ and $c$ are the benefit and cost for cooperation, respectively,}  
  that of the $n$-players public goods game is
  $(1-\frac{c}{n})/(c-1)$, where $c$ is the public good multiplier, and
  that of the Traveler's Dilemma game is 
  \gs{$\frac{1}{2}(b-1)$, where $b$ is the bonus.}
  Finally, the
  selfishness level of Cournot competition (an example of an infinite
  ordinal potential game), Tragedy of the Commons, and Bertrand
  competition is infinite.
\end{abstract}



\begin{flushright}
{\emph{The intelligent way to be selfish is \\ to work for the
welfare of others} \\ Dalai-Lama\footnote{\cite[p.~109]{Bow04}.}}
\end{flushright}

\section{Introduction}


The discrepancy in strategic games between the social welfare in a
Nash equilibrium and in a social optimum has been long recognized by
the economists. One of the flagship examples is Cournot competition, a
strategic game involving firms that simultaneously choose the
production levels of a homogeneous product.  The payoff functions in
this game describe the firms' profit in the presence of some
production costs, under the assumption that the price of the product
depends negatively on the total output.  It is well-known (see,
e.g.,~\citeR[pp.~174--175]{JR11}) that the price in the social optimum
is strictly higher than in the Nash equilibrium, which shows that the
competition between the producers of a product drives its price down.

In computer science the above discrepancy led to the introduction of the
notions of the \emph{price of anarchy} \cite{papa-kouts} and the
\emph{price of stability} \cite{SS03,ADK+04} that measure the ratio
between the social welfare in a worst and, respectively, a best Nash
equilibrium and a social optimum. This originated a huge research
effort aiming at determining both ratios for specific strategic games
that possess (pure) Nash equilibria.

These two notions are \emph{descriptive} in the sense that they assess the existing situation. 
Said differently, these notions quantify the discrepancy between the social welfare in a Nash equilibrium and a social optimum given the initial payoff functions. In contrast, we propose a notion that is \emph{normative} in the sense that it explains how to change these payoff functions to resolve such a discrepancy. Intuitively, we are asking the question how much of the social welfare needs to be added to the players' payoff functions so that their individual preferences can bring them to an optimal outcome for the society. 
On an abstract level, the approach that we propose here is related to one proposed by \citeA[p.~134]{Axelrod84:Evolution}, in chapter ``How to Promote Cooperation'', from where we cite: ``An excellent way to promote cooperation in a society is to teach people to care about the welfare of others.''

Our approach draws on the concept of \emph{altruistic games} (see, e.g., \citeR{Ledyard95}, and more recently \citeR{MM07}). In these games each player's payoff is modified by adding a positive fraction $\alpha$ of the social welfare in the considered joint strategy to the original payoff. 
The \bfe{selfishness level} of a game is defined as the infimum over all $\alpha \ge 0$ for which such a modification yields that a social optimum is realized in a pure Nash equilibrium.
The underlying property is monotonic in the sense that if for some $\alpha \ge 0$ 
a social optimum is a pure Nash equilibrium, then it is also the case for every $\beta \ge \alpha$.



Intuitively, the selfishness level of a game can be viewed as a
measure of the players' willingness to cooperate. A low selfishness
level indicates that the players are open to align their interests in
the sense that a small share of the social welfare is sufficient to
motivate them to choose a social optimum.  In contrast, a high
selfishness level suggests that the players are reluctant to cooperate
and a large share of the social welfare is needed to stimulate
cooperation among them. An infinite selfishness level means that
cooperation cannot be achieved through such means.

Notions like the price of stability and the price of anarchy were developed to measure the \emph{quality} of equilibria. In contrast, our notion of the selfishness level can be regarded as a measure of \emph{willingness to cooperate}. In general, these notions are incomparable (as we will also argue formally) and provide different insights into the underlying game.

Our main motivation for analyzing the selfishness level of strategic games is to gain a deeper understanding of the characteristics that influence the players' willingness to cooperate. 
As it turns out, for several games studied in this paper the selfishness level provides such insights.
To illustrate this point, we briefly elaborate on our findings for the public goods game and the fair cost sharing game.

In the public goods game there are $n$ players who want to contribute to a public good. Every player $i$ chooses an amount $s_i \in [0,b]$ that he contributes. A central authority collects all individual contributions,
multiplies their sum by $c > 1$ (for simplicity we assume here that $n
\ge c$) and distributes the resulting amount evenly among all players.
The payoff of player $i$ is thus $p_i(s) := b - s_i + \frac{c}{n}
\sum_{j} s_j$.
In the (unique) Nash equilibrium, every player attempts to ``free
ride'' by contributing $0$ to the public good (which is a dominant
strategy), while in the social optimum every player contributes the
full amount of $b$. As we will show, the selfishness level of this
game is $(1-\frac{c}{n})/(c-1)$. This bound suggests that the
temptation to free ride (i) increases as the number of players grows
and (ii) decreases as the parameter $c$ increases. Both phenomena were
observed by experimental economists, (see, e.g., \gs{the discussion in} \citeR[Section III.C.2]{Ledyard95}).
In comparison, the price of stability (which coincides with the price of anarchy) for this game is $c$. 

In a fair cost sharing game every player $i$ chooses a facility from a set of facilities $S_i \subseteq E$ available to him (for simplicity we discuss here only the case where players choose a single facility). The cost $c_e$ of every used facility $e \in E$ is shared evenly among the players using it.  
As we will prove, the selfishness level of this game is $\max\{0, \frac12 c_{\max}/c_{\min} - 1\}$, where $c_{\max}$ and $c_{\min}$ refer to the largest and smallest cost of a facility, respectively. Our analysis therefore reveals a threshold phenomenon which also makes sense intuitively: 
In order to motivate cooperation among the players it is crucial to convince the players having access to a facility with cost $c_{\min}$ to adhere to a social optimum. If $c_{\max} \le 2 c_{\min}$ this is easy because in a social optimum each such player either shares the cost of a facility $e$ with $c_e \ge c_{\min}$ with at least one other player or uses a facility of cost $c_{\min}$ exclusively by himself. Thus, it is in the self-interest of each player to cooperate and choose a social optimum in this case.
If $c_{\max} > 2 c_{\min}$ these players are reluctant to cooperate and the fraction of the social welfare that needs to be offered to them to incite cooperation grows proportionally to $c_{\max}/c_{\min}$. \citeA{ADK+04} showed that the price of stability and the price of anarchy of this game are $H_n$ and $n$, respectively, where $n$ denotes the number of players.\footnote{$H_n$ denotes the $n$th Harmonic number.} So these measures depend on the number of players. In contrast, our notion reveals a dependency on the discrepancy between the costs of the facilities.

\gs{
A large body of literature in experimental economics indicates that players have a tendency to cooperate in social dilemmas like the Prisoner's dilemma, the Traveler's dilemma or the public goods game, even though such behavior is ruled out by standard game-theoretic analysis. Several studies suggest that the willingness to cooperate depends on certain parameters of the underlying game (like group-size, magnitude of payoffs, etc.); see, e.g., \citeA{IW88,C+96,GH01,BCN05,D+08}. For example, \citeauthor{D+08}~observe that in the Prisoner's dilemma the willingness to cooperate increases with the ratio of cost over benefit for cooperation.
We therefore study the selfishness level of parameterized versions of these games. Our findings show that the selfishness level also exhibits a dependency on certain parameters of the game.}

In this paper, we define the selfishness level by taking pure Nash
equilibrium as the solution concept. This is in line with how the
price of anarchy and price of stability {were defined originally}
\cite{papa-kouts,SS03,ADK+04}. However, the definition applies equally well
  to other solution concepts and other forms of games.
We discuss these matters in the final section. 

\subsection{Our Contributions}

The main contributions presented in this paper are as follows: 

\begin{enumerate}

\item We introduce (in Section~\ref{sec:selfish}) the notion of selfishness level of a game, derive some basic properties and elaborate on some connections to other efficiency measures and models of altruism.

In particular, we show that the selfishness level of a game is unrelated to the price of stability and the price of anarchy. Moreover, the selfishness level is invariant under positive linear transformations of the payoff functions. We also show that our model is equivalent to other models of altruism that have been studied before. As a consequence, our bounds on the selfishness level directly transfer to these alternative models. 

\item We derive (in Section~\ref{sec:charac}) a characterization result that allows us to determine the selfishness level of a strategic game.

Our characterization shows that the selfishness level is determined by the maximum \emph{appeal factor} of unilateral profitable deviations from specific social optima, which we call \emph{stable}. As a result, we can focus on deviations from these stable social optima only. Intuitively, the appeal factor of a single player deviation refers to the ratio of the gain in his payoff over the resulting loss in social welfare. 

\item We use (in Section \ref{sec:examples}) our characterization result to analyze the selfishness level of several classical strategic games. 

The games that we study are fundamental and often used to illustrate the consequences of selfish behavior and the effects of competition.  A summary of our results is given in Table~\ref{tab:summary}.
In particular, we derive explicit bounds on the selfishness level of fair cost sharing games and congestion games with linear delay functions. The obtained bounds depend on specific parameters of the underlying game, which we find informative. We also show that these bounds are tight for certain instances.

\item We also show (in Section~\ref{sec:conclusions}) that our selfishness level notion naturally extends to other solution concepts and other types of games, for instance mixed Nash equilibria and extensive games.

\end{enumerate}

\begin{table}
\begin{center}
\renewcommand{\arraystretch}{1.2}
\noindent
\begin{tabular}{|p{8cm}|p{5cm}|}
\hline
\textbf{Game} & \textbf{Selfishness level} \\
\hline\hline
Ordinal potential games & finite \\
Weakly acyclic games & $\infty$ \\
Fair cost sharing games (singleton) & $\max\{0, \frac12 \frac{c_{\max}}{c_{\min}} - 1\}^{\dagger}$ \\
Fair cost sharing games (integer costs)  &   $\max\{0, \frac12 L c_{\max} - 1\}^{\dagger}$ \\
Linear congestion games (singleton) & 
$\max\{0, \frac12 \frac{\Delta_{\max} - \Delta_{\min}}{(1 - \discr_{\max}) a_{\min}} - \frac12\}^{\dagger}$ \\
Linear congestion games (integer coefficients) & $\max\{0, \frac12 (L\Delta_{\max} - \Delta_{\min} - 1)\}^{\dagger}$ \\
Prisoner's Dilemma for $n$ players & \gs{$\frac{c}{b(n-1)-c}$$^{\dagger}$} \\
Public goods game & $\max\{0, \frac{1-\frac{c}{n}}{c-1}\}^{\dagger}$ \\
Traveler's dilemma & \gs{$\frac12 (b-1)$$^{\dagger}$} \\
Cournout competition & $\infty$ \\
Tragedy of the commons & $\infty$ \\
Bertrand competition & $\infty$ \\
\hline
\end{tabular}
\caption{\label{tab:summary}
Selfishness level of the games studied in this paper.  \hfil \break
$^\dagger$ see Section~\ref{sec:examples} for the definitions of the respective parameters of the games.}
\end{center}
\end{table}

\subsection{Related Work}

There are only few articles in the algorithmic game theory literature
that study the influence of altruism in strategic games
\cite{CKKKP10,CKKS11,CK08,EMAN10,HS09}. In these works, altruistic
player behavior is modeled by altering each player's perceived payoff
in order to account also for the welfare of others. The models differ
in the way they combine the player's individual payoff with the
payoffs of the other players. All these studies are descriptive in the
sense that they aim at understanding the impact of altruistic behavior
on specific strategic games.

Closest to our work are the articles by \citeA{EMAN10} and by \citeA{CKKS11}.
\citeauthor{EMAN10} study the inefficiency of equilibria in
network design games (which constitute a special case of the cost 
sharing games considered here) with altruistic (or, as they call it, 
socially-aware) players. As we do here, they define each player's cost
function as his individual cost plus $\alpha$ times the social cost.
They derive lower and upper bounds on the price of anarchy and the
price of stability, respectively, of the modified game. In
particular, they show that the price of stability is at most $(H_n +
\alpha)/(1+\alpha)$, where $n$ is the number of players.

\citeA{CKKS11} introduce a framework to study the \emph{robust price of anarchy}, 
which refers to the worst-case inefficiency of other solution concepts such as coarse 
correlated equilibria (see \citeR{Rou09}) of altruistic extensions of strategic games. In their model, 
player $i$'s perceived cost is a convex combination of $(1-\gamma_i)$ times
his individual cost plus $\gamma_i$ times the social cost, where
$\gamma_i \in [0,1]$ is the altruism level of player $i$. If all players have
a uniform altruism level $\gamma_i = \gamma$, this model relates to the
one we consider here by setting $\alpha = \gamma/(1-\gamma)$
{(see Section~\ref{sec:model-relations} for details).}
Although not being the main focus of the paper, the authors
also provide upper bounds of $2/(1+\gamma)$ and $(1-\gamma) H_n + \gamma$
on the price of stability for linear congestion games and fair cost
sharing games, respectively.%

Note that in all three cases {mentioned above} the price of stability approaches 1 as
$\alpha$ goes to $\infty$. This seems to suggest that the
selfishness level of these games is $\infty$. However, this is not the case as our analyses reveal.

Two other models of altruism were proposed in the literature. \citeA{CK08} define the perceived cost of a player as $(1-\beta)$ times his individual cost plus $\beta/n$ times the social cost, where $\beta \in [0,1]$. \citeA{CKKKP10} define the perceived cost of player $i$ as $(1-\delta)$ times his individual cost plus $\delta$ times the sum of the costs of all other players (i.e., excluding player $i$), where $\delta \in [0,1]$. Also these two models can be shown to be equivalent to our model using simple transformations {(see Section~\ref{sec:model-relations} for details).}

Subsequently, we mention a few related approaches that are normative.
Conceptually, our selfishness level notion is related to the \emph{Stackelberg threshold} introduced by \citeA{SW09} (see also \citeR{KS09}). The authors consider network routing games in which a fraction of $\beta \in [0,1]$ of the flow is first routed centrally and the remaining flow is then routed selfishly. The Stackelberg threshold refers to the smallest value of $\beta$ that is needed to improve upon the social cost of a Nash equilibrium flow.

In a related paper, \citeA{HS09} study the minimum number, termed the \emph{optimal stability threshold}, of (pure) altruists that are needed in a congestion game to induce a Nash equilibrium as a social optimum. They show that this number can be computed in polynomial time for singleton congestion games.

In network congestion games, researchers studied the effect of imposing tolls on the edges of the network in order to reduce the inefficiency of Nash equilibria (see, e.g.,~\citeR{Beckmann:1956}).
From a high-level perspective, these approaches can also be regarded as normative. 

\gs{Recently, \citeA{C13} proposed a new normative approach to measure incentive for cooperation in symmetric games in which there is a tension between selfish and altruistic behavior. The solution concept is a pure Nash equilibrium of a transformed game in which the strategies are certain mixed strategic of the original game. These strategies depend on the incentive and risk of deviating from cooperation in the original game. Strikingly, Capraro's conclusions about the influence of the parameters in the Prisoner's Dilemma, Traveler's Dilemma and the public goods game are consistent with ours.}

There are several other papers that propose notions allowing to assess the \emph{stability} of Nash equilibria. We mention a few of them below.
\citeA{CKS11} study the inefficiency of approximate Nash equilibria in congestion games. In a \emph{$(1+\varepsilon)$-approximate} Nash equilibrium the cost of each player is at most $(1+\varepsilon)$ times the cost he experiences in every unilateral deviation. The authors derive (almost) tight bounds on the price of stability and the price of anarchy for linear (non-atomic and atomic) congestion games as a function of $\varepsilon$. In particular, they obtain a bound of $\min\{1, (1+\sqrt{3})/(\varepsilon + \sqrt{3})\}$ on the price of stability for atomic linear congestion games.
In this context, an alternative notion to assess the stability of Nash equilibria that comes to one's mind is to consider the smallest $\varepsilon \ge 0$ for which a social optimum is realized as a $(1+\varepsilon)$-Nash equilibrium. Note that the above bound implies that such an $\varepsilon$ is at most $1$ for linear congestion games.
We comment on this idea in more detail in Section~\ref{subsec:future}.

\citeA{AEY09} consider the problem of incentivizing players to participate in socially desirable matchings by adding \emph{switching costs} to player deviations. In their model, the additional cost that a player incurs by changing his strategy accounts for an $\varepsilon$ fraction of his individual cost. Adopting this viewpoint, the authors study the inefficiency of $(1+\varepsilon)$-approximate stable matchings. They derive bounds on the price of stability and the price of anarchy of $(1+\varepsilon)$-approximate stable matchings as a function of $\varepsilon \ge 0$. Related to this work is the article of \citeA{BMM10} who study the problem of computing an optimal matching having a minimum number of blocking pairs.

Furthermore, \citeA{BBM09} study the impact of advertising strategies to players in order to induce them to select more efficient equilibria. More precisely, in their model an authority first proposes a strategy to each player which is then accepted by each player with probability $\alpha$. Each accepting player adheres to the proposed strategy and all remaining players play a best response (assuming that the strategies of the accepting players are fixed). In a final step all players follow a best response dynamics until a Nash equilibrium is reached. The authors analyze the inefficiency of the resulting equilibria for fair cost sharing games, machine scheduling games and party affiliation games. In particular, for fair cost sharing games they show that the expected cost of the resulting equilibrium is at most a factor $O(\log n/\alpha)$ away from a social optimum.

\section{Selfishness Level}
\label{sec:selfish}

In this section, we formally introduce our notion of selfishness level, establish some properties and relate it to other notions of altruism.

\subsection{Definition}

A \bfe{strategic game} (in short, a game) $G = (N, \{S_i\}_{i \in N}, \{ p_i \}_{i \in N})$ is given by a set $N = \{1, \dots, n\}$ of $n > 1$ players, a non-empty set of \bfe{strategies} $S_i$ for every player $i \in N$, and a \bfe{payoff function} $p_i$ for every player $i \in N$ with $p_i : S_1 \times \dots \times S_n \myra \mathbb{R}$.
The players choose their strategies simultaneously and every player 
$i \in N$ aims at choosing a strategy $s_i \in S_i$ so as to maximize his individual payoff $p_i(s)$, where $s = (s_1, \dots, s_n)$.

We call $s \in S_1 \times \dots\times S_n$ a \bfe{joint strategy} and
denote its $i$th element by $s_i$. We
denote $(s_1, \dots, s_{i-1}, s_{i+1}, \dots, s_n)$ by $s_{-i}$ and 
similarly with $S_{-i}$. 
Further, we write $(s'_i, s_{-i})$ for $(s_1, \dots, s_{i-1}, s'_i, s_{i+1}, \dots, s_n)$, 
where we assume that $s'_i \in S_i$.  
Sometimes, when focusing on player $i$ we write $(s_i, s_{-i})$ instead of $s$.


A joint strategy $s$ is a \bfe{Nash equilibrium} if for all $i \in \{1,
\ldots, n\}$ and $s'_i \in S_i$, $p_i(s_i, s_{-i}) \geq p_i(s'_i, s_{-i})$.
Further, given a joint strategy $s$ we call the sum 
$SW(s) := \sum_{i = 1}^{n} p_i(s)$ the \bfe{social welfare} of $s$.  When the
social welfare of $s$ is maximal we call $s$ a \bfe{social optimum}.

We shall also consider a `cost' variant of the games in which we use
the cost functions, written as $c_i$, instead of the payoff functions
$p_i$. In such a setup the objective of each player is
to minimize his costs, so a joint strategy $s$ is a Nash equilibrium
if for all $i \in \{1, \ldots, n\}$ and $s'_i \in S_i$, $c_i(s_i,
s_{-i}) \leq c_i(s'_i, s_{-i})$. Further, instead of the social welfare
one considers the \bfe{social cost} of $s$, defined as $SC(s) :=
\sum_{i = 1}^{n} c_i(s)$.

Given a strategic game $G := (N, \{S_i\}_{i \in N}, \{p_i\}_{i \in N})$
and $\alpha \geq 0$ we define the game $G(\alpha) := (N, \{S_i\}_{i \in N}, \{r_i\}_{i \in N})$
by putting $r_i(s) := p_i(s) + \alpha SW(s)$. 
So when $\alpha > 0$ the payoff of each player in the $G(\alpha)$ game
depends on the social welfare of the players.  $G(\alpha)$ is then an
\bfe{altruistic version} of the game $G$.

Suppose now that for some $\alpha \geq 0$ a pure Nash equilibrium of
$G(\alpha)$ is a social optimum of $G(\alpha)$. Then we say that $G$
is \bfe{$\alpha$-selfish}.  
We define the \bfe{selfishness level} of $G$ as
\begin{equation}\label{eq:sel-lev}
\inf \{ \alpha \in \mathbb{R}_+ \; \mid \; \mbox{$G$ is $\alpha$-selfish} \}.
\end{equation}
Here we adopt the convention that the infimum of an empty set is $\infty$. 
Further, we stipulate that the selfishness level of $G$ is denoted by $\alpha^+$ iff the selfishness level of $G$ is 
$\alpha \in \mathbb{R}_+$ but $G$ is \emph{not} $\alpha$-selfish (equivalently, the infimum does not belong to the set).
We show below (Theorem~\ref{not:inf}) that pathological infinite games exist
for which the selfishness level is of this kind; none of the other studied games is of this type. 

We give some remarks before we proceed.
\begin{enumerate}

\item The above definitions refer to strategic games in which each player
$i$ maximizes his payoff function $p_i$ and the social welfare of a
joint strategy $s$ is given by $SW(s)$. These definitions obviously
apply to the case when we use cost functions and the social cost.

\item Other definitions of an altruistic version of a game are conceivable and, depending on the underlying application, might seem more natural than the one we use here. However, we show in Section~\ref{sec:model-relations} that our definition is equivalent to several other models of altruism that have been proposed in the literature.

\item The selfishness level refers to the smallest $\alpha$ such that
    \emph{some} Nash equilibrium in $G(\alpha)$ is also a social
    optimum. Alternatively, one might be interested in the smallest
    $\alpha$ such that \emph{every} Nash equilibrium in $G(\alpha)$
    corresponds to a social optimum. However, as explained in
    Section~\ref{subsec:future}, this alternative notion is not always
    very meaningful.

\item The definition extends in the obvious way to other solution concepts (e.g., mixed or correlated equilibria) and other forms of games (e.g., subgame perfect equilibria in extensive games). We briefly comment on these extensions in Section~\ref{sec:conclusions}.

\end{enumerate}

Note that the social welfare of a joint strategy $s$ in $G(\alpha)$
equals $(1 + \alpha n) SW(s)$, so the social optima of $G$ and
$G(\alpha)$ coincide.  Hence we can replace in the definition
of an $\alpha$-selfish game
the reference to a social optimum of $G(\alpha)$ by one to a social
optimum of $G$. This is what we shall do in the proofs below.

Intuitively, a low selfishness level means that the share of the
social welfare needed to induce the players to choose a social optimum
is small. This share can be viewed as an `incentive' needed to realize
a social optimum.
Let us illustrate this definition on various simple examples.

\begin{example}\label{exa:prisoners-dilemma}
\textbf{Prisoner's Dilemma}
\begin{center}
\begin{game}{2}{2}
       & $C$    & $D$\\
$C$   & \gs{$1,\phantom{-}1$}   & \gs{$-1,\phantom{-}2$} \\
$D$   & \gs{$2,-1$}             & \gs{$\phantom{-}0,\phantom{-}0$}
\end{game}
\hspace*{1cm}
\begin{game}{2}{2}
       & $C$    & $D$\\
$C$   & \gs{$3,3$}   & \gs{$0,3$}\\
$D$   & \gs{$3,0$}   & \gs{$0,0$}
\end{game}
\end{center}

Consider the Prisoner's Dilemma game $G$ (on the left) and the resulting game
$G(\alpha)$ for $\alpha = 1$ (on the right).
In the latter game the social optimum, $(C,C)$, is also a Nash equilibrium.
One can easily check that for $\alpha < 1$, $(C,C)$ is also a social
optimum of $G(\alpha)$ but not a Nash equilibrium.  So the selfishness
level of this game is 1.
\end{example}

\begin{example}
\textbf{Battle of the Sexes}
\begin{center}
\begin{game}{2}{2}
       & $F$    & $B$\\
$F$   &$2,1$   &$0,0$\\
$B$   &$0,0$   &$1,2$
\end{game}
\end{center}

Here each Nash equilibrium is also a social optimum, so 
the selfishness level of this game is 0.
\end{example}

\begin{example}\label{exa:matching}
\textbf{Matching Pennies} 

\begin{center}
\begin{game}{2}{2}
      & $H$    & $T$\\
$H$   &$\phantom{-}1,-1$   &$-1,\phantom{-}1$\\
$T$   &$-1,\phantom{-}1$   &$\phantom{-}1,-1$
\end{game}
\end{center}

Since the social welfare of each joint strategy is 0, for each $\alpha$
the game $G(\alpha)$ is identical to the original game in which no
Nash equilibrium exists. So the selfishness level of this game is $\infty$.
More generally, the selfishness level of a constant sum game is 0 
if it has a Nash equilibrium and otherwise it is $\infty$.
\end{example}

\begin{example}\label{exa:bad}
\textbf{Game with a bad Nash equilibrium} \\
The following game results from equipping each
player in the Matching Pennies game with a third strategy $E$ (for
edge):

\begin{center}
\begin{game}{3}{3}
      & $H$     & $T$       & $E$ \\
$H$   & $\phantom{-}1,-1$  & $-1,\phantom{-}1$    & $-1,-1$ \\
$T$   & $-1,\phantom{-}1$  & $\phantom{-}1,-1$    & $-1,-1$ \\
$E$   & $-1,-1$            & $-1,-1$              & $-1,-1$ 
\end{game}
\end{center}

Its unique Nash equilibrium is $(E,E)$. It is easy to check that the selfishness level of this game is $\infty$.
(This is also an immediate consequence of Theorem~\ref{thm:characterization}~(iii) below.)
\end{example}

\begin{example}\label{exa:no}
\textbf{Game with no Nash equilibrium} \\
Consider a game $G$ on the left and the resulting game
$G(\alpha)$ for $\alpha = 1$ on the right.

\begin{center}
\begin{game}{2}{2}
       & $C$    & $D$\\
$C$   &$2,2$   &$2,0$\\
$D$   &$3,0$   &$1,1$
\end{game}
\hspace*{1cm}
\begin{game}{2}{2}
       & $C$    & $D$\\
$C$   &$6,6$   &$4,2$\\
$D$   &$6,3$   &$3,3$
\end{game}\end{center}

The game $G$ has no Nash equilibrium, while in the game  $G(1)$ the social optimum, $(C,C)$, is also a Nash equilibrium.
As in the Prisoner's Dilemma game one can easily check that for $\alpha < 1$, $(C,C)$ is also a social
optimum of $G(\alpha)$ but not a Nash equilibrium.  So the selfishness level of the game $G$ is 1.
\end{example}

\subsection{Properties}\label{subsec:properties}

Recall that, given a finite game $G$ that has a Nash equilibrium, its \bfe{price of stability} is the ratio $SW(s)/SW(s')$ where $s$ is a social optimum and $s'$ is a Nash equilibrium with the highest social welfare in $G$. The  \bfe{price of anarchy} is defined as the ratio $SW(s)/SW(s')$ where $s$ is a social optimum and $s'$ is a Nash equilibrium with the lowest social welfare in $G$. 

So the price of stability of $G$ is 1 iff its selfishness level  is 0. However, in general there is no relation between these two notions.  The following observation also shows that the selfishness level of a finite game can be an arbitrary real number.

\begin{theorem}\label{not:alpha}
For every finite $\alpha > 0$ and $\beta > 1$ there is a finite game whose selfishness level is $\alpha$
and whose price of stability is $\beta$.
\end{theorem}

\begin{proof}
Consider the following generalized form, which we denote by $PD(\alpha, \beta)$, of the Prisoner's Dilemma game $G$ with $x = \frac{\alpha}{\alpha+1}$:
\begin{center}
\begin{game}{2}{2}
       & $C$    & $D$\\
$C$   & $\phantom{x+1}1,1$   & $\phantom{\frac{1}{\beta}}0,x+1\phantom{\frac{1}{\beta}}$\\
$D$   & $\phantom{1}x+1,0$   & $\phantom{0}\frac{1}{\beta},\frac{1}{\beta}\phantom{x+1}$
\end{game}
\end{center}

In this game and in each game $G(\gamma)$ with $\gamma \ge 0$,
$(C,C)$ is the unique social optimum. To compute the
selfishness level we need to consider a game $G(\gamma)$ and
stipulate that $(C,C)$ is its Nash equilibrium.
This leads to the inequality
$1 + 2 \gamma \geq (\gamma +1) (x + 1)$,
from which it follows that $\gamma \geq \frac{x}{1-x}$, i.e., $\gamma \geq \alpha$.
So the selfishness level of $G$ is $\alpha$. Moreover, its price of stability is $\beta$,
since $(D,D)$ is the only Nash equilibrium.
\end{proof}

The notion of the selfishness level is invariant under simple payoff
transformations. It is a direct consequence of the following
observation, where given a game $G$ and a value $a$ we denote by $G +
a$ (respectively, $a G$) the game obtained from $G$ by adding to each
payoff function the value $a$ (respectively, by multiplying each
payoff function by $a$).

\begin{proposition}\label{prop:transformations}
Consider a game $G$ and $\alpha \geq 0$.
\begin{enumerate}[(i)]
\item For every $a$, $G$ is $\alpha$-selfish iff $G + a$ is $\alpha$-selfish.
\item For every $a > 0$, $G$ is $\alpha$-selfish iff $a G$ is $\alpha$-selfish.
\end{enumerate}
\end{proposition}
\begin{proof}
(i) It suffices to note that $r[a]_i(s) = r_i(s) + \alpha a n + a$,
where $r_i$ and $r[a]_i$ are the payoff functions 
of player $i$ in the games $G(\alpha)$ and  $(G + a)(\alpha)$. 
So for every joint strategy $s$
\begin{itemize}\itemsep0pt
\item $s$ is a Nash equilibrium of $G(\alpha)$ iff it is a Nash equilibrium of $(G + a)(\alpha)$,
\item $s$ is social optimum of $G(\alpha)$ iff it is a social optimum of $(G + a)(\alpha)$.
\end{itemize}

(ii) It suffices to note that for every $a > 0$, $r[a]_i(s) =  a r_i(s)$,
where this time  $r[a]_i$ is the payoff function
of player $i$ in the game $(a G)(\alpha)$, and argue as above.
\end{proof}

Proposition~\ref{prop:transformations} implies that the selfishness
level is invariant under the game transformations of the form $t(G):=
a G + b$, where $a > 0$.  This is in contrast to the notions of the
price of anarchy and the price of stability that are invariant only
under the game transformations of the form $t(G):= a G$, where $a > 0$.

Note that the selfishness level is not invariant under a
multiplication of the payoff functions by a value $a \leq 0$.  Indeed,
for $a = 0$ each game $a G$ has the selfishness level 0. For $a < 0$
take the game $G$ from Example~\ref{exa:bad} whose selfishness level is $\infty$.
In the game $a G$ the joint strategy $(E,E)$ is both a Nash equilibrium and a social optimum,
so the selfishness level of  $a G$ is 0.

{The above proposition also allows us to frame the notion of selfishness level in the following way. 
Suppose the original $n$-player game $G$ is given to a game designer who has a fixed
budget of $SW(s)$ for each joint strategy $s$ and that the selfishness level of $G$ is 
$\alpha < \infty$. How should the game designer then distribute the budget of $SW(s)$ for each joint strategy $s$ among the players such that the resulting game has a Nash equilibrium that coincides with a social optimum?
By scaling $G(\alpha)$ by the factor $a:=1/(1 + \alpha n)$ we ensure that for each joint strategy $s$ its social welfare in the original game $G$ and in $aG(\alpha)$ is the same.
Using Proposition~\ref{prop:transformations}, we conclude that $\alpha$ is the smallest non-negative real such that $a G(\alpha)$ has a Nash equilibrium that is a social optimum.
The game $a G(\alpha)$ can then be viewed as the
intended transformation of $G$. That is, each payoff function $p_i$ of
the game $G$
is transformed into the payoff function
\[
r_i(s) := \frac{1}{1 + \alpha n} p_i(s) + \frac{\alpha}{1 + \alpha n} SW(s).
\]
}

Let us return now to the `borderline case' of the selfishness level that we denoted by $\alpha^+$.
We have the following result.

\begin{theorem}\label{not:inf}
For every $\alpha \geq 0$ there exists a game whose selfishness level is $\alpha^+$.
\end{theorem}

\begin{proof} 
We first prove the result for $\alpha = 0$. That is, we show that there exists 
a game that is $\alpha$-selfish for every $\alpha > 0$, but is not 0-selfish.
To this end we use the games $PD(\alpha, \beta)$ defined in the proof of Theorem~\ref{not:alpha}.

We construct a strategic game $G = (N, \{S_i\}_{i \in N}, \{p_i\}_{i
  \in N})$ with two players $N = \{1, 2\}$ by combining, for an
arbitrary but fixed $\beta > 1$, infinitely many $PD(\alpha, \beta)$
games with $\alpha > 0$ as follows: For each $\alpha > 0$ we rename
the strategies of the $PD(\alpha, \beta)$ game to, respectively,
$C(\alpha)$ and $D(\alpha)$ and denote the corresponding payoff
functions by $p_i^\alpha$.
The set of strategies of each player $i \in N$ is $S_i = \{C(\alpha) \mid \alpha > 0\} \cup \{D(\alpha) \mid \alpha > 0\}$
and the payoff of $i$ is defined as 
\[
p_i(s_i,s_{-i}) := 
\begin{cases}
       p^{\alpha}_i(s_i,s_{-i})    & \text{if $\{s_{i}, s_{-i}\} \subseteq \{C(\alpha), \ D(\alpha)\}$ for some $\alpha > 0$} \\
        0       & \mathrm{otherwise.}
\end{cases}
\]


Every social optimum of $G$ is of the form $(C(\alpha), C(\alpha))$,
where $\alpha > 0$. (Note that we exploit that $\beta > 1$ here.) By
the argument given in the proof of Theorem~\ref{not:alpha},
$(C(\alpha), C(\alpha))$ with $\alpha > 0$ is a Nash equilibrium in
the game $G(\alpha)$ because the deviations from $C(\alpha)$ to a
strategy $C(\gamma)$ or $D(\gamma)$ with $\gamma \neq \alpha$ yield a
payoff of 0. Thus, $G$ is $\alpha$-selfish for every $\alpha > 0$.
Finally, observe that $G$ is not $0$-selfish because every Nash
equilibrium of $G$ is of the form $(D(\alpha), D(\alpha))$, where
$\alpha > 0$.

To deal with the general case we prove two claims that
are of independent interest.

\begin{claim}\label{cla:1}
  For every game $G$ and
  $\alpha \geq 0$ there is a game $G'$ such that $G'(\alpha) = G$.
\end{claim}
\begin{proof}
We define the payoff of player $i$ in the game $G'$ by
\[
p'_i(s) := p_i(s) - \frac{\alpha}{1+ n \alpha} SW(s),
\]
where $p_i$ is his payoff in the game $G$.
Denote by $SW'(s)$ the social welfare of a joint strategy $s$ in the game $G'$
and by $r'_i$ the payoff function of player $i$ in the game $G'(\alpha)$.
Then 
\begin{align*}
r'_i(s)  & = p'_i(s) + \alpha SW'(s) \\
& =  p_i(s) - \frac{\alpha}{1+ n \alpha} SW(s) + \alpha \left(SW(s) - \frac{n \alpha}{1+ n \alpha} SW(s)\right) \\
& =  p_i(s) + \left(\alpha -\frac{\alpha}{1+ n \alpha} - \frac{n \alpha^2}{1+ n \alpha}\right) SW(s) \\
& =  p_i(s).
\end{align*}
\end{proof}

\begin{claim}\label{cla:2}
For every game $G$ and $\alpha, \beta \geq 0$
\[
G(\alpha + \beta) = G(\alpha)\left(\frac{\beta}{1 + n \alpha}\right).
\]
\end{claim}
\begin{proof}
Denote by $SW'(s)$ the social welfare of a joint strategy $s$ in the game $G(\alpha)$,
by $p_i, r_i$ and $r'$ the payoff functions of player $i$ in the games $G$, $G(\alpha)$,
and $G(\alpha)(\frac{\beta}{1 + n \alpha})$.
Then
\[
r_i(s) := p_i(s) + \alpha SW(s),
\]
so
\begin{align*}
r'_i(s)  & =  r_i(s) + \frac{\beta}{1 + n \alpha} SW'(s) \\
& =  p_i(s) + \alpha SW(s) + \frac{\beta}{1 + n \alpha} (SW(s) + n \alpha SW(s)) \\
& =  p_i(s) + \left(\alpha + \frac{\beta}{1 + n \alpha} + \frac{\beta n \alpha}{1 + n \alpha}\right) SW(s) \\
& =  p_i(s) + (\alpha + \beta) SW(s),
\end{align*}
which proves the claim.
\end{proof}

To prove the general case fix $\alpha \geq 0$ and $\beta > 0$ and take
a game $G$ whose selfishness level is $0^+$.  By Claim~\ref{cla:1}
there is a game $G'$ such that $G'(\alpha) = G$.  Then $G'$ is not
$\alpha$-selfish, since $G$ is not $0$-selfish.

Further, by Claim~\ref{cla:2}
\[
G'(\alpha + \beta) = G'(\alpha)\left(\frac{\beta}{1 + n \alpha}\right) = G\left(\frac{\beta}{1 + n \alpha}\right).
\] 
But by its choice the game $G$ is $\frac{\beta}{1 + n \alpha}$-selfish, so
$G'$ is $(\alpha+\beta)$-selfish, which concludes the proof.
\end{proof}

\subsection{Alternative Definitions}
\label{sec:model-relations}

Our definition of the selfishness level depends on the way the altruistic versions of the original game are defined. Three other models of altruism were proposed in the literature. 
As before, let $G := (N, \{S_i\}_{i \in N}, \{p_i\}_{i \in N})$ be a strategic game. Consider the following four definitions of altruistic versions of $G$:

\begin{description}
\item[Model A \cite{EMAN10}:] 
For every $\alpha \ge 0$, $G(\alpha) := (N, \{S_i\}_{i \in N}, \{r^{\alpha}_i\}_{i \in N})$ with 
\begin{equation}\label{eq:modelA}
r^\alpha_i(s) = p_i(s) + \alpha SW(s) \quad \forall i \in N.
\end{equation}

\item[Model B \cite{CK08}:] 
For every $\beta \in [0,1]$, $G(\beta) :=  (N, \{S_i\}_{i \in N}, \{r^{\beta}_i\}_{i \in N})$ with  
\begin{equation}\label{eq:modelB}
r^{\beta}_i(s) = (1- \beta) p_i(s) + \frac{\beta}{n} SW(s) \quad \forall i \in N.
\end{equation}

\item[Model C \cite{CKKS11}:] For every $\gamma \in [0,1]$, $G(\gamma) :=  (N, \{S_i\}_{i \in N}, \{r^{\gamma}_i\}_{i \in N})$ with
\begin{equation}\label{eq:modelC}
r^{\gamma}_i(s) = (1- \gamma) p_i(s) + \gamma SW(s) \quad \forall i \in N.
\end{equation}

\item[Model D \cite{CKKKP10}:] For every $\delta \in [0,1]$, $G(\delta) :=  (N, \{S_i\}_{i \in N}, \{r^{\delta}_i\}_{i \in N})$ with
\begin{equation}\label{eq:modelD}
r^{\gamma}_i(s) = (1- \delta) p_i(s) + \delta(SW(s) - p_i(s)) \quad \forall i \in N.
\end{equation}
\end{description}

Our selfishness level notion for Model A extends to Models B, C and D in the obvious way: 
We say that $G$ is \bfe{$\beta$-selfish} for some $\beta \in [0,1]$ iff a pure Nash equilibrium of the altruistic version $G(\beta)$ is also a social optimum. The \bfe{selfishness level of $G$ with respect to Model B} is then defined as the infimum over all $\beta \in [0,1]$ such that $G$ is $\beta$-selfish. The respective notions for Models C and D are defined analogously. 

The following theorem shows that the selfishness level of a game with respect to Models A, B, C and D relate to each other via simple transformations. {(Note that for Model D this transformation only applies for $\delta \in [0, \frac12]$.)}

\begin{theorem}
Consider a strategic game $G := (N, \{S_i\}_{i \in N}, \{p_i\}_{i \in N})$ and its altruistic versions defined according to Models A, B, C and D above.
\begin{enumerate}[(i)]
\item $G$ is $\alpha$-selfish with $\alpha \in \mathbb{R}_+$ iff $G$ is $\beta$-selfish with $\beta = \frac{\alpha n}{1+ \alpha n} \in [0,1]$.
\item $G$ is $\alpha$-selfish with $\alpha \in \mathbb{R}_+$ iff $G$ is $\gamma$-selfish with $\gamma = \frac{\alpha}{1+\alpha} \in [0,1]$.
\item $G$ is $\alpha$-selfish with $\alpha \in \mathbb{R}_+$ iff $G$ is $\delta$-selfish with $\delta = \frac{\alpha}{1+ 2 \alpha} \in [0,\frac12]$.
\end{enumerate}
\end{theorem}
\begin{proof}
We prove the following more general claim. Fix $x, y > 0$. For every $\lambda \in [0, \frac{1}{x}]$, define $G(\lambda) :=  (N, \{S_i\}_{i \in N}, \{r^{\lambda}_i\}_{i \in N})$ with  
\begin{equation}\label{eq:ModelL}
r_i^\lambda(s) = (1-x \lambda) p_i(s) + \frac{\lambda}{y} SW(s).
\end{equation}
We show that $G$ is $\alpha$-selfish for $\alpha \ge 0$ iff $G$ is $\lambda$-selfish for $\lambda = \frac{\alpha y}{1+ \alpha xy} \in [0, \frac{1}{x}]$.

By substituting $\lambda = \frac{\alpha y}{1+ \alpha xy}$ in \eqref{eq:ModelL}, we obtain 
$$
r_i^{\lambda}(s) = \frac{1}{1+\alpha xy} p_i(s) + \frac{\alpha}{1+\alpha xy} SW(s) = \frac{1}{1 + \alpha xy} r_i^{\alpha}(s).
$$
As a consequence, since $\frac{1}{1+ \alpha xy} > 0$ for every $\alpha \ge 0$ the pure Nash equilibria and social optima, respectively, of $G(\lambda)$ and $\frac{1}{1+\alpha xy} G(\alpha)$ coincide. Thus, $G$ is $\lambda$-selfish iff $\frac{1}{1+\alpha xy} G$ is $\alpha$-selfish. Also, it follows from Proposition~\ref{prop:transformations} that $\frac{1}{1 + \alpha xy} G$ is $\alpha$-selfish iff $G$ is $\alpha$-selfish. 

Further, note that
$$
\lim_{\alpha \rightarrow \infty} \frac{\alpha y}{1+ \alpha xy} = \frac{1}{x} \left( 1 - \lim_{\alpha \rightarrow \infty} \frac{1}{1 + \alpha xy}  \right)= \frac{1}{x}.
$$
That is, the selfishness level of $G$ with respect to Model A is $\infty$ iff the selfishness level of $G$ with respect to $G(\lambda)$ is $\frac{1}{x}$.

Now, (i) follows from the above with $x = 1$ and $y = n$, (ii) follows with $x = y = 1$ and (iii) follows with $x = 2$ and $y = 1$.
\end{proof}

\section{A Characterization Result}
\label{sec:charac}

We now characterize the games with a finite selfishness level.
To this end we shall need the following notion.
We call a social optimum $s$ \bfe{stable} if for all $i \in N$ and $s'_i \in S_i$
the following holds: 
\[
\text{if $(s'_i, s_{-i})$ is a social optimum, then $p_i(s_i, s_{-i}) \geq p_i(s'_i, s_{-i})$.}
\]
In other words, a social optimum is stable if no player is better off
by unilaterally deviating to another social optimum.  

It will turn out that in order to determine the selfishness level of a
game we need to consider deviations from its stable social optima.
Consider a deviation $s'_i$ of player $i$ from a stable social optimum
$s$. If player $i$ is better off by deviating to $s'_i$, then by
definition the social welfare decreases, i.e., $SW(s_i, s_{-i}) -
SW(s'_i, s_{-i}) > 0$. If in the original game this decrease is small,
while the gain for player $i$ is large, then strategy $s'_i$ is an
attractive and socially acceptable option for player $i$.  We define
player $i$'s \bfe{appeal factor} of strategy $s'_i$ given the social
optimum $s$ as
\[
\af_i(s'_i, s) := \frac{p_i(s'_i, s_{-i}) - p_i(s_i, s_{-i})}{SW(s_i, s_{-i}) - SW(s'_i, s_{-i})}.
\]

In what follows we shall characterize the selfishness level in terms
of bounds on the appeal factors of profitable deviations from a stable social optimum.
First, note the following properties of social optima.

\begin{lemma} \label{lem:characterization}
Consider a strategic game $G := (N, \{S_i\}_{i \in N}, \{p_i\}_{i \in N})$ and $\alpha \geq 0$.
\begin{enumerate}[(i)]
\item If $s$ is both a Nash equilibrium of $G(\alpha)$ and a social optimum of $G$,
then $s$ is a stable social optimum of $G$.
\item If $s$ is a stable social optimum of $G$, then
$s$ is a Nash equilibrium of $G(\alpha)$ iff 
for all $i \in N$ and $s'_i \in U_i(s)$,
$\alpha \geq \af_i(s'_i, s)$,
where 
\begin{equation}\label{contour-set}
U_i(s) := \{s'_i \in S_i\; \mid \mbox{ $p_i(s'_i, s_{-i}) > p_i(s_i, s_{-i})$}\}. 
\end{equation}
\end{enumerate}
\end{lemma}
The set $U_i(s)$, with the ``$>$'' sign replaced by ``$\geq$'',
is called an \emph{upper contour set} (see, e.g., \citeR[p.~193]{Rit02}).
Note that if $s$ is a stable social optimum, then 
$s'_i \in U_i(s)$ implies that $SW(s_i, s_{-i}) > SW(s'_i, s_{-i})$.

\begin{proof}
(i)
Suppose that $s$ is both a Nash equilibrium of $G(\alpha)$ and a social optimum of $G$.
Consider some joint strategy $(s'_i, s_{-i})$ that is a social optimum.  By the definition of
a Nash equilibrium
\[
p_i(s_i, s_{-i}) + \alpha SW(s_i, s_{-i}) \geq p_i(s'_i, s_{-i}) + \alpha SW(s'_i, s_{-i}),
\]
so $p_i(s_i, s_{-i}) \geq p_i(s'_i, s_{-i})$, as desired.

(ii) Suppose that $s$ is a stable social optimum of $G$.
Then $s$ is a Nash equilibrium of $G(\alpha)$ 
iff for all $i \in N$ and $s'_i \in S_i$
\begin{equation}
p_i(s_i, s_{-i}) + \alpha SW(s_i, s_{-i}) \geq p_i(s'_i, s_{-i}) + \alpha SW(s'_i, s_{-i}).
\label{equ:geq}  
\end{equation}

If $p_i(s_i, s_{-i}) \geq p_i(s'_i, s_{-i})$, then
(\ref{equ:geq})  holds for all $\alpha \geq 0$ since $s$ is a social optimum.
If $p_i(s'_i, s_{-i}) > p_i(s_i, s_{-i})$, then,
since $s$ is a stable social optimum of $G$, we have
$SW(s_i, s_{-i}) > SW(s'_i, s_{-i})$.

So (\ref{equ:geq}) holds for all $i \in N$ and $s'_i \in S_i$
iff
\[
\alpha \geq \frac{p_i(s'_i, s_{-i}) - p_i(s_i, s_{-i})}{SW(s_i, s_{-i}) - SW(s'_i, s_{-i})} = \af_i(s'_i, s)
\]
holds for all $i \in N$ and $s'_i \in U_i(s)$.
\end{proof}

This leads us to the following result.

\begin{theorem} \label{thm:characterization}
  Consider a strategic game $G := (N, \{S_i\}_{i \in N}, \{p_i\}_{i \in N})$.
  \begin{enumerate}[(i)]
  \item The selfishness level of $G$ is finite iff a stable social optimum $s$ exists for which
$
\alpha(s) := \sup_{i \in N,\, s'_i \in U_i(s)} \af_i(s'_i, s)
$
is finite.

\item 
If the selfishness level of $G$ is finite, then it equals
$\min_{s \in SSO} \alpha(s)$, where SSO is the set of stable social optima.

\item If $G$ is finite, then its selfishness level is finite iff it
  has a stable social optimum.  In particular, if $G$ has a unique
  social optimum, then its selfishness level is finite.

\item If $\beta > \alpha \geq 0$ and $G$ is  $\alpha$-selfish, then 
$G$ is  $\beta$-selfish.
\end{enumerate}

\end{theorem}

\begin{proof}
(i) and (iv) follow by Lemma \ref{lem:characterization},
(ii) by (i) and Lemma \ref{lem:characterization}, and (iii) by (i).
\end{proof}

Using the above theorem we now exhibit a class of games for $n$ players for which the
selfishness level is unbounded. In fact, the following more general result holds.

\begin{theorem}
For each function $f: \mathbb{N} \myra \mathbb{R}_+$
there exists a class of games for $n$ players, where $n > 1$, such that
the selfishness level of a game for $n$ players equals $f(n)$.   
\end{theorem}

\begin{proof}
Assume $n > 1$ players and that each player has two strategies, $1$ and $0$. 
Denote by $\textbf{1}$ the
joint strategy in which each strategy equals $1$ and by 
$\textbf{1}_{-i}$ the
joint strategy of the opponents of player $i$ in which each entry equals $1$.
The payoff for each player $i$ is defined as follows:
\[
p_i(s) := \begin{cases}
          0                 & \mathrm{if} \  s = \textbf{1} \\
          f(n)                 & \mathrm{if} \ s_i = 0 \text{ and }  \forall j < i,\  s_j = 1 \\
          - \frac{f(n)+1}{n-1} & \mathrm{otherwise.}
                    \end{cases}
\]
So when $s \neq \textbf{1}$,  $p_i(s) = f(n)$ if $i$ is the smallest index of a player with $s_i = 0$
and otherwise $p_i(s) =  - \frac{f(n)+1}{n-1}$.
Note that $SW(\textbf{1}) = 0$ and $SW(s) = -1$ if $s \neq \textbf{1}$.
So $\textbf{1}$ is a unique social optimum.

We have $p_i(0, \textbf{1}_{-i}) - p_i(\textbf{1}) = f(n)$
and $SW(\textbf{1}) - SW(0, \textbf{1}_{-i}) = 1$.
So by Theorem \ref{thm:characterization}~(ii)
the selfishness level equals $f(n)$.
\end{proof}

\section{Examples}
\label{sec:examples}

We now use the above characterization result to determine or compute an upper bound on
the selfishness level of some selected games.  First, we exhibit a
well-known class of games (see \citeR{MS96}) for which the selfishness
level is finite.  

\subsection{Ordinal Potential Games}

Given a game $G := (N, \{S_i\}_{i \in N}, \{p_i\}_{i \in N})$, a
function $P: S_1 \times \dots\times S_n \myra \mathbb{R}$ is called an
\bfe{ordinal potential function} for $G$ if for all $i \in N$, $s_{-i} \in S_{-i}$ and $s_i, s'_i \in S_i$,
$p_i(s_i, s_{-i}) > p_i(s'_i, s_{-i})$ iff $P(s_i, s_{-i}) > P(s'_i, s_{-i})$. 
A game that possesses an ordinal potential function is called an
\bfe{ordinal potential game}.

\begin{theorem}\label{thm:pot}
  Every finite ordinal potential game has a finite selfishness level.
\end{theorem}

\begin{proof}
Each social optimum with the largest potential is a stable social optimum.
So the claim follows by Theorem \ref{thm:characterization}~(ii).
\end{proof}

In particular, every finite congestion game (see \citeR{Ros73}) 
has a finite selfishness level. 
We shall derive explicit bounds for two special cases of these games in Sections~\ref{exa:cost-sharing} and \ref{exa:congestion}.

\subsection{Weakly Acyclic Games}

Given a game $G := (N, \{S_i\}_{i \in N},
\{p_i\}_{i \in N})$, a \bfe{path} in $S_1 \times \dots\times S_n$ is a
sequence $(s^1, s^2, \dots)$ of joint strategies such that for every $k
> 1$ there is a player $i$ such that $s^k = (s'_i, s^{k-1}_{-i})$ for
some $s'_i \neq s^{k-1}_{i}$ (see, e.g., \citeR{MS96}). 
A path is called an \bfe{improvement
  path} if it is maximal and for all $k > 1$, $p_i(s^k) >
p_i(s^{k-1})$, where $i$ is the player who deviated from $s^{k-1}$.  A
game $G$ has the \bfe{finite improvement property} (\bfe{FIP}) if
every improvement path is finite. 
A game $G$ is called \bfe{weakly acyclic} if for every joint strategy there exists a finite improvement path that starts at it (see, e.g., \citeR{Mil96,You93}).

Finite games that have the FIP coincide with the ordinal potential games.
So by Theorem~\ref{thm:pot} these games have a finite selfishness level. 
In contrast, the selfishness level of a weakly acyclic game can be infinite.
Indeed, the following game is easily seen to be weakly acyclic:
\begin{center}
\begin{game}{3}{3}
      & $H$                           & $T$                             & $E$ \\
$H$   & $\phantom{-51}1,-1$           & $\phantom{1}-1,\phantom{-}1$    & $\phantom{1}-1,-0.5$ \\
$T$   & $\phantom{-}-1, \phantom{-}1$  & $\phantom{-55}1,-1$              & $\phantom{1}-1,-0.5$ \\
$E$   & $-0.5,-1$                     & $-0.5,-1$                       & $-0.5,-0.5$ \\
\end{game}
\end{center}
Yet, on the account of Theorem \ref{thm:characterization}~(iii), its
selfishness level is infinite.

\subsection{Fair Cost Sharing Games}\label{exa:cost-sharing}

In this and the next subsection we consider cost-minimization instead of payoff-maximization games. Recall that in these games each player $i$ wants to minimize his individual cost function $c_i$ and that the social cost is defined as $SC(s) = \sum_{i} c_i(s)$.

In a fair cost sharing game (see, e.g., \citeR{ADK+04}) players allocate facilities and share the
cost of the used facilities in a fair manner. Formally, a fair cost
sharing game is given by $G = (N, E, \{S_i\}_{i \in N}, \{c_e\}_{e \in
  E})$, where $N = \{1, \dots, n\}$ is the set of players, $E$ is the
set of facilities, $S_i \subseteq 2^E$ is the set of facility subsets
available to player $i$, and $c_e \in \mathbb{R}_+$ is the cost of
facility $e \in E$. It is called a \emph{singleton} cost sharing game if for every $i \in N$ and for every $s_i \in S_i$: $|s_i| = 1$. 
For a joint strategy $s \in S_1 \times \dots \times S_n$ let $x_e(s)$ be the number of players using facility $e \in E$, i.e., $x_e(s) = |\{i \in N \mid e \in s_i\}|$.
The cost of a facility $e \in E$ is evenly shared among the players using it. That is, the cost of player $i$ is defined as $c_i(s) = \sum_{e \in s_i} c_e/x_e(s)$. 

We first consider singleton cost sharing games.
Let $c_{\max} = \max_{e \in E} c_e$ and $c_{\min} = \min_{e \in E} c_e$ refer to the maximum and minimum costs of the facilities, respectively. 

\sloppy
\begin{proposition}\label{prop:cost-sharing}
The selfishness level of a singleton cost sharing game is at most {$\max\{0, \frac12 \frac{c_{\max}}{c_{\min}} - 1\}$}. Moreover, this bound is tight.
\end{proposition}

\fussy
This result should be contrasted with the price of stability of $H_n$ and the price of anarchy of $n$ for  cost sharing games \cite{ADK+04}. 
Cost sharing games admit an exact potential function and thus by Theorem~\ref{thm:pot} their selfishness level is finite. However, as the tight example given in the proof of Proposition~\ref{prop:cost-sharing} below shows, the selfishness level can be arbitrarily large (as $c_{\max}/c_{\min} \rightarrow \infty$) even for $n = 2$ players and two facilities.

In order to prove Proposition~\ref{prop:cost-sharing}, we first derive an expression of the appeal factor for arbitrary fair cost sharing games, which we then specialize to singleton cost sharing games to prove the claim. 

Let $s$ be a stable social optimum. Note that $s$ exists by Theorem~\ref{thm:characterization}~(iii) and Theorem \ref{thm:pot}.
Because we consider a cost minimization game here the appeal factor of player $i$ is defined as
\begin{equation}\label{eq:af-cost}
\af_i(s'_i, s) := \frac{c_i(s_i, s_{-i}) - c_i(s'_i, s_{-i})}{SC(s'_i, s_{-i}) - SC(s_i, s_{-i})}
\end{equation}
and the condition in Theorem~\ref{thm:characterization}~(i) reads
$\alpha(s) := \max_{i \in N,\,  s'_i \in U_i(s)} \af_i(s'_i, s)$,
where $U_i(s) := \{s'_i \in S_i \; \mid \; c_i(s'_i, s_{-i}) < c_i(s_i, s_{-i})\}$. 

Fix some player $i$ and let $s' = (s'_i, s_{-i})$ for some $s'_i \in U_i( s)$. We use $x_e$ and $x'_e$ to refer to $x_e(s)$ and $x_e(s')$, respectively. 
Note that
$$
x'_e = 
\begin{cases}
x_e + 1 & \text{if $e \in s'_i \setminus s_i$}, \\
x_e - 1 & \text{if $e \in s_i \setminus s'_i$}, \\
x_e & \text{otherwise.} 
\end{cases}
$$

We have 
\begin{align}\label{eq:diff-cost}
c_i(s) - c_i(s'_i, s_{-i}) 
 & = \sum_{e \in s_i \setminus s'_i} \frac{c_e}{x_e} - \sum_{e \in s'_i \setminus s_i} \frac{c_e}{x_e+1}.
\end{align}
Further, it is not difficult to verify that 
\begin{align}\label{eq:diff-SC}
SC(s'_i, s_{-i}) - SC(s) 
= \sum_{e \in s'_i \setminus s_i:\, x_e = 0} c_e - \sum_{e \in s_i \setminus s'_i:\, x_e = 1} c_e.
\end{align}
Thus,
\begin{align}\label{eq:af-cs}
\af_i(s'_i, s)
& = \frac{\sum_{e \in s_i \setminus s'_i:\, x_e \ge 2} \frac{c_e}{x_e} - \sum_{e \in s'_i \setminus s_i:\, x_e \ge 1} \frac{c_e}{x_e+1}}{\sum_{e \in s'_i \setminus s_i:\, x_e = 0} c_e - \sum_{e \in s_i \setminus s'_i:\, x_e = 1} c_e} - 1.
\end{align}

We use the above to prove Proposition~\ref{prop:cost-sharing}.

\begin{proof}[Proof of Proposition~\ref{prop:cost-sharing}]
{Let $s$ be a stable social optimum (which exists by Theorem~\ref{thm:characterization}~(iii) and Theorem \ref{thm:pot}).
If $U_i(s) = \emptyset$ for every $i \in N$ then the selfishness level is $0$ by Theorem~\ref{thm:characterization}~(ii). 
Otherwise, there is some player $i \in N$ with $U_i(s) \neq \emptyset$. 
Recall that in a singleton cost sharing game, each player's strategy set consists of singleton facility sets. 
Let $s_i = \{e\}$ and $s'_i = \{ e' \}$ be the singleton sets of the facilities chosen by player $i$ in $s$ and in $s' = (s'_i, s_{-i})$ with $s'_i \in U_i(s)$.} Clearly, $e \neq e'$.

Note that $SC(s'_i, s_{-i}) - SC(s)$ must be positive because $s'_i \in U_i(s)$ and thus \eqref{eq:diff-SC} implies that $x_{e'} = 0$. Therefore, \eqref{eq:diff-cost} reduces to 
$c_i(s) - c_i(s'_i, s_{-i})  = c_e/x_e - c_{e'}$.
If $x_e = 1$ then $c_e > c_{e'}$ because $s'_i \in U_i(s)$. But this is a contradiction to the assumption that $SC(s'_i, s_{-i}) - SC(s) =  c_{e'} - c_e > 0$. Thus $x_e \ge 2$. {Note that this also implies that $c_e > 2 c_{e'}$ and thus $c_{\max} > 2 c_{\min}$.}

Using \eqref{eq:af-cs}, we obtain 
$$
\af_i(s'_i, s) = \frac{\frac{c_e}{x_e}}{c_{e'}} - 1 \le \frac{1}{2} \frac{c_{\max}}{c_{\min}} - 1.
$$
The claim now follows by Theorem~\ref{thm:characterization}~(ii).

The following example shows that this bound is tight. Suppose $N = \{1, 2\}$, $E = \{e_1, e_2\}$, $S_1 = \{ \{e_1\} \}$, $S_2 = \{ \{e_1\}, \{e_2\} \}$, $c_{e_1} = c_{\max}$ and $c_{e_2} = c_{\min}$ {with $c_{\max} > 2 c_{\min}$}. The joint strategy $s = (\{e_1\}, \{e_1\})$ is the unique social optimum with $SC(s) = c_{\max}$ and $c_2(s) = c_{\max}/2$. Suppose player $2$ deviates to $s'_2 = \{e_2\}$. Then $SC(s'_2, s_1) = c_{\max} + c_{\min}$ and $c_2(s'_2, s_1) = c_{\min}$. Thus $\af_i(s'_2, s) = (\tfrac12 c_{\max} - c_{\min})/c_{\min} = \frac12 c_{\max}/c_{\min} - 1$.
\end{proof}

The following example shows that a bound similar to the one above, i.e., bounding the selfishness level in terms of the ratio $c_{\max}/c_{\min}$, does not hold for arbitrary fair cost sharing games. In particular, it shows that the minimum difference between any two costs of facilities (here $\varepsilon$) must enter a bound of the selfishness level for arbitrary fair cost sharing games.

\begin{example}\label{ex:cost-diff}
Let $N = \{1, 2\}$, $E = \{e_1, e_2, e_3\}$, $S_1 = \{ \{e_1\} \}$, $S_2 = \{ \{e_1, e_3\}, \{e_2\} \}$, $c_{e_1} = c_{\max}$, $c_{e_2} = c_{\min} + \varepsilon$ for some $\varepsilon > 0$ and $c_{e_3} = c_{\min}$. The joint strategy $s = (\{e_1\}, \{e_1, e_3\})$ is the unique social optimum with $SC(s) = c_{\max} + c_{\min}$ and $c_2(s) = c_{\max}/2 + c_{\min}$. Suppose player $2$ deviates to $s'_2 = \{e_2\}$. Then $SC(s'_2, s_1) = c_{\max} + c_{\min} + \varepsilon$ and $c_2(s'_2, s_1) = c_{\min} + \varepsilon$. Thus $\af_i(s'_2, s) = (\tfrac12 c_{\max} - \varepsilon)/\varepsilon = \frac12 c_{\max}/\varepsilon - 1$, which approaches $\infty$ as $\varepsilon \rightarrow 0$.
\end{example}

We next derive a bound for arbitrary fair cost sharing games with non-negative integer costs.
Let $L$ be the maximum number of facilities that any player can choose, i.e., $L := \max_{i \in N, s_i \in S_i} |s_i|$. 

\begin{proposition}\label{prop:cost-sharing-general}
The selfishness level of a fair cost sharing game with non-negative integer costs is at most {$\max\{0, \frac12 L c_{\max} - 1\}$.} Moreover, this bound is tight.
\end{proposition}

\begin{proof}
Let $s$ be a stable social optimum. As in the proof of Proposition~\ref{prop:cost-sharing}, if $U_i(s) = \emptyset$ for every $i \in N$ then the selfishness level is $0$ by Theorem~\ref{thm:characterization}~(ii). Otherwise, there is some player $i \in N$ with $U_i(s) \neq \emptyset$. Let $s' = (s'_i, s_{-i})$ for some $s'_i \in U_i(s)$.
Note that the denominator of the appeal factor in \eqref{eq:af-cs} is at least $1$ because {$s$ is stable}, $s'_i \in U_i(s)$ and $c_e \in \mathbb{N}$ for each $e \in E$. Thus
\begin{align*}
\af_i(s'_i, s) & = \frac{\sum_{e \in s_i \setminus s'_i:\, x_e \ge 2} \frac{c_e}{x_e} - \sum_{e \in s'_i \setminus s_i:\, x_e \ge 1} \frac{c_e}{x_e+1}}{\sum_{e \in s'_i \setminus s_i:\, x_e = 0} c_e - \sum_{e \in s_i \setminus s'_i:\, x_e = 1} c_e} - 1 \\
& \le  \sum_{e \in s_i \setminus s'_i:\, x_e \ge 2} \frac{c_e}{x_e} - 1 \le \frac12 L  c_{\max} - 1.
\end{align*}
The claim follows by Theorem \ref{thm:characterization} (ii). 

The following example shows that the bound is tight. Suppose we are given $L$ and $c_{\max}$. Let $N = \{1, \dots, n\}$ and $E = \{e_1, \dots, e_n\}$ where $n = L+1$. Define $S_i = \{ \{ e_i \} \}$ for every $i \in N \setminus \{n\}$ and $S_n = \{ \{e_1, \dots, e_{n-1}\}, \{e_n\}\}$. Let $c_{e_i} = c_{\max}$ for every $i \in N \setminus \{n\}$ and $c_{e_n} = 1$. The joint strategy $s = (\{e_1\}, \dots, \{e_{n-1}\}, \{e_1, \dots, e_{n-1}\})$ is the unique social optimum with $SC(s) = (n-1)c_{\max}$ and $c_n(s) = (n-1) c_{\max}/2$. Suppose player $n$ deviates to $s'_n = \{ e_n\}$. Then $SC(s'_n, s_{-n}) = (n-1) c_{\max} + 1$ and $c_n(s'_n, s_{-n}) = 1$. Thus $\af_i(s'_n, s) = \frac12 (n-1) c_{\max} - 1 = \frac12 L c_{\max} - 1$.
\end{proof}

\begin{remark}\label{rem:scaling}
We can bound the selfishness level of a fair cost sharing game with non-negative rational costs $c_e \in \mathbb{Q}_+$ for every facility $e \in E$ by using Proposition~\ref{prop:cost-sharing-general} and the following  scaling argument: Simply scale all costs to integers, e.g., by multiplying them with the least common multiplier $q \in \mathbb{N}$ of the denominators.
Note that this scaling does not change the selfishness level of the game by Proposition~\ref{prop:transformations}. However, it does change the maximum facility cost and thus $q$ enters the bound. 
Also note that this scaling implicitly takes care of the effect observed in Example~\ref{ex:cost-diff}:
Suppose that $c_{\max}$ and $c_{\min}$ are integers and $\epsilon = 1/q$ for some $q \in \mathbb{N}$. Then all costs are multiplied by $q$ and Proposition~\ref{prop:cost-sharing-general} yields a (non-tight) bound of $q c_{\max} - 1 = c_{\max}/\epsilon-1$ on the selfishness level, which approaches $\infty$ as $q \rightarrow \infty$.
\end{remark}

\subsection{Linear Congestion Games}\label{exa:congestion}

In a congestion game $G := (N, E, \{S_i\}_{i \in N}, \{d_e\}_{e \in E})$ we are given a set of players $N = \{1, \ldots, n\}$, a set of facilities $E$ with a delay function $d_e : \mathbb{N} \rightarrow \mathbb{R}_+$ for every facility $e \in E$, and a strategy set $S_i \subseteq 2^E$ for every player $i \in N$. 
For a joint strategy $s \in S_1 \times \dots \times S_n$, define $x_e(s)$ as the number of players using facility $e \in E$, i.e., $x_e(s) = |\{i \in N \mid e \in s_i\}|$. The goal of a player is to minimize his individual cost $c_i(s) = \sum_{e \in s_i} d_e(x_e(s))$. 

Here we call a congestion game \emph{symmetric} if there is some common strategy set $S \subseteq 2^E$ such that $S_i = S$ for all $i$. It is \emph{singleton} if every strategy $s_i \in S_i$ is a singleton set, i.e., for every $i \in N$ and for every $s_i \in S_i$, $|s_i| = 1$.
In a \emph{linear} congestion game, the delay function of every facility $e \in E$ is of the form $d_e(x) = a_e x + b_e$, where $a_e, b_e \in \mathbb{R}_+$ are non-negative real numbers. 

We first derive a bound on the selfishness level for symmetric singleton linear congestion games. 
As it turns out, a bound similar to the one for singleton cost sharing games does not extend to symmetric singleton linear congestion games. Instead, the crucial insight here is that the selfishness level depends on the \emph{discrepancy} between facilities in a stable social optimum. 
We make this notion more precise.

Let $s$ be a stable social optimum and let $x_e$ refer to $x_e(s)$. Define the \emph{discrepancy} between two facilities $e$ and $e'$ {with $a_e + a_{e'} > 0$} under $s$ as 
\begin{equation}\label{eq:disc}
\discr(x_e, x_{e'}) = \frac{2a_e x_e + b_e}{a_e + a_{e'}} - \frac{2 a_{e'} x_{e'} + b_{e'}}{a_e + a_{e'}}.
\end{equation}

We show below that $\discr(x_e, x_{e'}) \in [-1, 1]$.
Define $\discr_{\max}(s)$ as the maximum discrepancy between any two facilities $e$ and $e'$ under $s$ with $a_e + a_{e'} > 0$ and $\discr(x_e, x_{e'}) < 1$; more formally, let 
$$
\discr_{\max}(s) = \max_{e, e' \in E} \{ \discr(x_e, x_{e'}) \; | \; a_e + a_{e'} > 0 \text{ and } \discr(x_e, x_{e'}) < 1 \}.
$$
Let $\discr_{\max}$ be the maximum discrepancy over all stable social optima, i.e., $\discr_{\max} = \max_{s \in SSO} \discr_{\max}(s)$.
Further, let $\Delta_{\max} :=  \max_{e \in E} (a_e + b_e)$ and $\Delta_{\min} := \min_{e \in E} (a_e + b_e)$.
Moreover, let $a_{\min}$ be the minimum non-zero coefficient of a latency function, i.e., $a_{\min} = \min_{e \in E: a_e > 0} a_e$. 

\begin{proposition}\label{prop:lin-cg}
The selfishness level of a symmetric singleton linear congestion game is at most 
$$
\max\left\{0, \frac12 \frac{\Delta_{\max} - \Delta_{\min}}{(1-\discr_{\max}) a_{\min}} - \frac12\right\}.
$$
Moreover, this bound is tight.
\end{proposition}

We first prove that the discrepancy between two facilities is bounded:

\begin{claim}\label{claim:range}
Let $s$ be a social optimum and $e, e' \in E$ be two facilities with $a_e + a_{e'} > 0$. 
Then the discrepancy between $e$ and $e'$ under $s$ satisfies $\discr(x_e, x_{e'}) \in [-1, 1]$. 
\end{claim}
\begin{proof}
Let $t = x_e + x_{e'}$ be the total number of players on facilities $e$ and $e'$ under $s$. Note that since $s$ is a social optimum and strategy sets are symmetric, $t$ is distributed among $x_e$ and $x_{e'}$ such that the social cost of these two facilities is minimized. Said differently, $x_e = x$ minimizes the function
$$
f(x, t) := a_e x^2 + b_e x + a_{e'} (t-x)^2 + b_{e'} (t-x).
$$
It is not hard to verify that the minimum of $f(x,t)$ (for fixed $t$) is attained at the (not necessarily integral) point
$$
\bar{x}_0 := \frac{2a_{e'} t - b_e + b_{e'}}{2(a_e + a_{e'})}.
$$
Because $f(x,t)$ is a parabola with its minimum at $\bar{x}_0$, the integral point $x_e$ that minimizes $f(x,t)$ is given by the point obtained by rounding $\bar{x}_0$ to the nearest integer. Let $x_e := \bar{x}_0 + \tfrac12 \discr$
be this point, where $\discr = \discr(x_e, x_{e'}) \in [-1, 1]$, and $x_{e'} = t - x_e$. Note that the choice of $\discr$ is unique, unless $\bar{x}_0$ is half-integral in which case $\discr \in \{-1, 1\}$.
Solving these equations for $\discr$ yields the definition in \eqref{eq:disc}. 
\end{proof}

\begin{proof}[Proof of Proposition~\ref{prop:lin-cg}]
Let $s$ be a stable social optimum. Note that $s$ exists by Theorem~\ref{thm:characterization}~(iii) and Theorem \ref{thm:pot}.
{If $U_i(s) = \emptyset$ for every $i \in N$ then the selfishness level is $0$ by Theorem~\ref{thm:characterization}~(ii). 
Otherwise, there is some player $i \in N$ with $U_i(s) \neq \emptyset$. 
Let $s' = (s'_i, s_{-i})$ for some $s'_i \in U_i( s)$.}
{We use $x_e$ and $x'_e$ to refer to $x_e(s)$ and $x_e(s')$ for every facility $e \in E$, respectively. Note that for every $e \in E$ we have}
\begin{equation}\label{eq:relx}
x'_e = 
\begin{cases}
x_e + 1 & \text{if $e \in s'_i \setminus s_i$}, \\
x_e - 1 & \text{if $e \in s_i \setminus s'_i$}, \\
x_e & \text{otherwise.} 
\end{cases}
\end{equation}
{Let $s_i = \{e\}$ and $s'_i = \{ e' \}$ be the sets of facilities chosen by player $i$ in $s$ and $s'$, respectively.
Exploiting \eqref{eq:relx}, we obtain }
\begin{equation}
c_i(s_i, s_{-i}) - c_i(s'_i, s_{-i}) 
= a_e x_e + b_e - a_{e'} (x_{e'} + 1) - b_{e'}. \label{eq:cg-diff} 
\end{equation}
Moreover, 
\begin{equation}
SC(s'_i, s_{-i}) - SC(s_i, s_{-i}) 
= a_{e'} (2x_{e'} + 1) + b_{e'} - a_e (2x_e-1) - b_e. \label{eq:cg-SC-diff}
\end{equation}
{Note that we have $c_i(s_i, s_{-i}) - c_i(s'_i, s_{-i}) > 0$ because $s'_i \in U_i(s)$ and by the definition of $U_i(s)$ in \eqref{contour-set}. Further, $SC(s'_i, s_{-i}) - SC(s_i, s_{-i}) > 0$ because $s$ is a stable social optimum and $s'_i \in U_i(s)$. Thus, it must hold that $a_e + a_{e'} > 0$; otherwise $a_e = a_{e'} = 0$ and \eqref{eq:cg-diff} and \eqref{eq:cg-SC-diff} yield a contradiction.}

Let $\discr = \discr(x_e, x_{e'})$ be the discrepancy between $e$ and $e'$ under $s$. Note that $\delta \in [-1, 1]$ by Claim~\ref{claim:range}. Using the definition of $\delta$ in \eqref{eq:disc}, we can rewrite \eqref{eq:cg-diff} and \eqref{eq:cg-SC-diff} as 
$$
c_i(s_i, s_{-i}) - c_i(s'_i, s_{-i}) 
= \tfrac12  (a_e + a_{e'}) \discr + \tfrac12 b_e - \tfrac12 b_{e'}  - a_{e'} 
$$
and
$$
SC(s'_i, s_{-i}) - SC(s_i, s_{-i}) 
=  (1 - \discr) (a_e + a_{e'}). 
$$
We conclude that $\delta \neq 1$. 

Thus,
\begin{align*}
\af_i(s'_i, s) 
& = \frac12 \cdot \frac{  (a_e + a_{e'})\discr + b_e - b_{e'} - 2a_{e'}}{(1 - \discr) (a_e + a_{e'})}
= \frac12 \cdot \frac{(a_e + b_e) - (a_{e'} + b_{e'})}{(1 - \discr) (a_e + a_{e'})} - \frac12 \\
& 
\le \frac12 \cdot \frac{\Delta_{\max} - \Delta_{\min}}{(1 - \discr_{\max}) a_{\min}} - \frac12.
\end{align*}
The claim now follows by Theorem \ref{thm:characterization}~(ii).

The following example shows that this bound is tight even for $n = 2$ players and two facilities. 
Let $N = \{1, 2\}$, $E = \{e, e'\}$ and $S_1 = S_2 = \{ \{e\}, \{e'\}\}$. 
Suppose we are given $\discr \in [0, 1)$ and $a_{e'} \in \mathbb{R}_+$. 
Define $d_{e}(x) = (2+\discr) a_{e'}$ and $d_{e'}(x) = a_{e'} x$. The joint strategy $s = (\{e\}, \{e'\})$ is the unique social optimum with $SC(s) = (3 +\discr) a_{e'}$. 
Further $c_1(s) = (2+\discr) a_{e'}$ and $c_2(s) = a_{e'}$. 
Suppose player $1$ deviates to $s'_1 = \{e'\}$. Then $SC(s'_1, s_2) = 4a_{e'}$ and $c_1(s'_1, s_2) = 2a_{e'}$. Thus $\af_i(s'_1, s) = \discr/(1-\discr)$, which matches precisely the upper bound given above. The case $\discr \in [-1, 0]$ is proven analogously.
\end{proof}

Observe that the selfishness level depends on the ratio $(\Delta_{\max}-\Delta_{\min})/a_{\min}$ and $1/(1-\discr_{\max})$. In particular, the selfishness level becomes arbitrarily large as $\discr_{\max}$ approaches $1$.

We next derive a bound for the selfishness level of arbitrary congestion games with linear delay functions and non-negative integer coefficients, i.e., $d_e(x) = a_e x + b_e$ with $a_e, b_e \in \mathbb{N}$ for every $e \in E$. Let $L$ be the maximum number of facilities that any player can choose, i.e., $L := \max_{i \in N, s_i \in S_i} |s_i|$.

\begin{proposition}\label{prop:lin-cg-general}
The selfishness level of a linear congestion game with non-negative integer coefficients is at most {$\max\{0, \frac12(L \Delta_{\max} - \Delta_{\min} - 1)\}$}. Moreover, this bound is tight.
\end{proposition}

For linear congestion games, the price of anarchy is known to be $\frac{5}{2}$ (see~\citeR{CK05,AAE05}). 
Our bound shows that the selfishness level depends on the maximum number of facilities in a strategy set and the magnitude of the coefficients of the delay functions.

\begin{proof}[Proof of Proposition~\ref{prop:lin-cg-general}]
Let $s$ be a stable social optimum. Note that $s$ exists by Theorem ~\ref{thm:characterization}~(iii) and Theorem \ref{thm:pot}. 
{If $U_i(s) = \emptyset$ for every $i \in N$ then the selfishness level is $0$ by Theorem~\ref{thm:characterization}~(ii). 
Otherwise, there is some player $i \in N$ with $U_i(s) \neq \emptyset$.}
Let $s' = (s'_i, s_{-i})$ for some $s'_i \in U_i( s)$. We use $x_e$ and $x'_e$ to refer to $x_e(s)$ and $x_e(s')$, respectively. 

Exploiting \eqref{eq:relx}, we obtain 
\begin{align*}
c_i(s_i, s_{-i}) - c_i(s'_i, s_{-i}) & 
 = \sum_{e \in s_i \setminus s'_i} (a_e x_e + b_e) - \sum_{e \in s'_i \setminus s_i} (a_e (x_e+1) + b_e).
\end{align*}
Similarly, 
\begin{align*}
SC(s'_i, s_{-i}) - SC(s_i, s_{-i}) 
& = \sum_{e \in s'_i \setminus s_i} (x_e + 1) (a_e (x_e + 1) + b_e) - x_e (a_e x_e + b_e) \\
& + \sum_{e \in s_i \setminus s'_i} (x_e - 1) (a_e (x_e - 1) + b_e) - x_e (a_e x_e + b_e) \\
& = \sum_{e \in s'_i \setminus s_i} (a_e (2x_e + 1) + b_e) - \sum_{e \in s_i \setminus s'_i} (a_e (2x_e-1) + b_e).
\end{align*}
Given a congestion vector $\vec x = (x_e)_{e \in E}$, define $P(\vec x) := \sum_{e \in s_i \setminus s'_i} (a_e x_e + b_e)$ and $Q(\vec x) := \sum_{e \in s'_i \setminus s_i} (a_e (x_e+1) + b_e)$. Note that $P(\vec x)$ and $Q(\vec x)$ are integers because $a_e, b_e \in \mathbb{N}$ for every facility $e \in E$. Note that with these definitions, $P(\vec 1) = \sum_{e \in s_i \setminus s'_i} (a_e + b_e)$ and $Q(\vec 0) = \sum_{e \in s'_i \setminus s_i} (a_e + b_e)$.
We have 
\[
\af_i(s'_i, s) 
= \frac{P(\vec x)-Q(\vec x)}{2Q(\vec x) -Q(\vec 0) - 2P(\vec x) + P(\vec 1)} .
\]
Because $s'_i \in U_i( s)$, we know that $P(\vec x) > Q(\vec x)$ and $2Q(\vec x) - Q(\vec 0) > 2P(\vec x) - P(\vec 1)$.
So we obtain 
$$
Q(\vec x) + 1 \le P(\vec x) \le Q(\vec x) + \tfrac12 (P(\vec 1) - Q(\vec 0) - 1).
$$
Exploiting these inequalities, we obtain
\begin{align*}
\af_i(s'_i, s)
& \le \frac12 (P(\vec 1) - Q(\vec 0) - 1) = \frac12 \bigg(\sum_{e \in s_i \setminus s'_i} (a_e + b_e) - \sum_{e \in s'_i \setminus s_i} (a_e + b_e) - 1 \bigg) \\
& \le \frac12 (|s_i \setminus s'_i| \cdot \Delta_{\max} - |s'_i \setminus s_i| \cdot \Delta_{\min} - 1).
\end{align*}
Note that $|s'_i \setminus s_i| \ge 1$; otherwise, $s'_i \subseteq
s_i$ and thus $SC(s'_i, s_{-i}) \le SC(s)$ which contradicts $s'_i \in
U_i(s)$.  The above expression is thus at most
$
\tfrac12(L \Delta_{\max} - \Delta_{\min} - 1).
$
The claim now follows by Theorem~\ref{thm:characterization}~(ii).

{The following example shows that this bound is tight. Fix $L$, $\Delta_{\max}$ and $\Delta_{\min}$  such that $(2n-1)\Delta_{\min} = L\Delta_{\max} + 1$ for some integer $n$.
Consider a congestion game with $N = \{1, \dots, n\}$ and $E = \{e_1, \dots, e_{L+1} \}$. Define $S_i = \{ \{e_{L+1}\} \}$ for every $i \in N \setminus \{n\}$ and $S_n = \{ \{e_1, \dots, e_{L} \}, \{e_{L+1}\}\}$. Let $d_{e_{L+1}}(x) = \Delta_{\min} x$ and $d_{e_i}(x) = \Delta_{\max}$ for every $i \in \{1, \dots, L\}$. 
For the joint strategy $s = (\{e_{L+1}\}, \dots, \{e_{L+1}\}, \{e_1, \dots, e_{L} \})$ we have $SC(s) = \Delta_{\min} (n-1)^2 + L \Delta_{\max}$ and $c_n(s) = L \Delta_{\max}$. If player $n$ deviates to $s'_n = \{ e_{L+1}\}$ we have $SC(s'_n, s_{-n}) = \Delta_{\min} n^2 = \Delta_{\min} (n-1)^2 + \Delta_{\min} (2n - 1)$ and $c_n(s'_n, s_{-n}) = \Delta_{\min} n$. Exploiting that $(2n-1)\Delta_{\min} = L\Delta_{\max} + 1$, we conclude that $SC(s) < SC(s'_n, s_{-n})$ and $c_n(s) > c_n(s'_n, s_{-n})$ (for $n \ge 3$). Thus, $s$ is a social optimum and $s'_n \in U_i(s)$. We obtain 
\begin{align*}
\af_n(s'_n, s) & = \frac{L\Delta_{\max} - \Delta_{\min} n}{\Delta_{\min}(2n-1) - L \Delta_{\max}}
= L\Delta_{\max} - \frac12 (L\Delta_{\max} +\Delta_{\min} + 1)\\
& = \frac12 (L\Delta_{\max} - \Delta_{\min} - 1).
\end{align*}
}
\end{proof}

\begin{remark}
We can use Proposition~\ref{prop:lin-cg-general} and the scaling argument outlined in Remark~\ref{rem:scaling} to derive bounds on the selfishness level of congestion games with linear delay functions and non-negative rational coefficients. 
\end{remark}

\subsection{Prisoner's Dilemma for $n$ Players}\label{exa:pri-n}

We assume that each player $i \in N = \{1, \dots, n\}$ has two strategies, $1$ (cooperate) and $0$ (defect).  
We put $p_i(s) := - cs_i + b \sum_{j \neq i} s_j$, where $b > c$. Intuitively, $b$ stands for the benefit of cooperation and $c$ for the cost of cooperation.

\begin{proposition}
The selfishness level of the $n$-players Prisoner's Dilemma game is $\frac{c}{b(n-1) - c}$.
\end{proposition}

Intuitively, this means that when the number of players in the  Prisoner's Dilemma game
increases, a smaller share of the social welfare is needed to
resolve the underlying conflict. The same observation holds for the value of the benefit.
That is, the `acuteness' of the dilemma diminishes
with the number of players and also diminishes when the value of the benefit grows.
The formal reason is that the appeal factor of each unilateral deviation from the social optimum
is inversely proportional to the number of players and inversely proportional to the benefit.

\begin{proof}
In this game $s = \vec{1}$ is the unique social optimum, with
for each $i \in N$, $p_i(s) = bn - (b + c)$ 
and $SW(s) = b n^2 - (b+c) n$.
Consider now the joint strategy $(s'_i, s_{-i})$ in which player $i$
deviates to the strategy $s'_i = 0$.
We have then $p_i(s'_i, s_{-i}) = bn - b$ and $SW(s'_i, s_{-i}) = b n^2 - (b+c) n + c - b(n-1)$.
Hence $\af_i(s'_i, s) = \frac{c}{b(n-1) - c}$.
The claim now follows by  Theorem~\ref{thm:characterization}~(ii).
\end{proof}

In particular, for $n = b = 2$ and $c = 1$ we get the original Prisoner's Dilemma game considered in
Example~\ref{exa:prisoners-dilemma} and as already argued there the selfishness level is then 1.

\subsection{Public Goods}\label{exa:public-goods}

We consider the public goods game with $n$ players. Every player $i \in N = \{1, \dots, n\}$ chooses an amount $s_i \in [0, b]$ that he contributes to a public good, where $b \in \mathbb{R}_+$ is the budget. The game designer collects the individual contributions of all players, multiplies their sum by $c > 1$ and distributes the resulting amount evenly among all players. The payoff of player $i$ is thus $p_i(s) := b - s_i + \frac{c}{n} \sum_{j \in N} s_j$. 

\begin{proposition}\label{prop:public-good}
The selfishness level of the $n$-players public goods game is $\max\big\{0, \frac{1 - \frac{c}{n}}{c-1}\big\}$.
\end{proposition}

In this game, every player has an incentive to ``free ride'' by contributing $0$ to the public good (which is a dominant strategy if $c \le n$). This is exactly as in the $n$-players Prisoner's Dilemma game \gs{(where defect is a dominant strategy if $c > 0$)}. However, the above proposition reveals that for fixed $c$, in contrast to the Prisoner's Dilemma game, this temptation becomes stronger as the number of players increases. Also, for a fixed number of players this temptation becomes weaker as $c$ increases.

\begin{proof}[Proof of Proposition~\ref{prop:public-good}]
Note that $SW(s) = bn + (c - 1)\sum_{i \in N} s_i$. The unique social optimum of this game is therefore $s = \vec{b}$ with $p_i(s) = cb$ for every $i \in N$ and $SW(s) = cbn$. Suppose player $i$ deviates from $s$ by choosing $s'_i \in [0, b)$. Then $p_i(s'_i, s_{-i}) = cb + (1 - \tfrac{c}{n})(b - s'_i)$.
Thus,
$$
p_i(s'_i, s_{-i}) - p_i(s) = (1 - \tfrac{c}{n})(b - s'_i)
\quad\text{and}\quad
SW(s) - SW(s'_i, s_{-i}) = (c-1) (b - s'_i).
$$
If $1 - \tfrac{c}{n} \le 0$ then $U_i(s) = \emptyset$ and the selfishness level is zero. Otherwise, $1 - \tfrac{c}{n} > 0$ and $U_i(s) = [0,b)$.
We conclude that in this case $\af_i(s'_i, s) = (1-\tfrac{c}{n})/(c-1)$ for every $s'_i \in U_i(s)$. 
The claim now follows by  Theorem~\ref{thm:characterization}~(ii).
\end{proof}

\subsection{Traveler's Dilemma}\label{exa:traveller}

This is a strategic game discussed by \citeA{Bas94} with two players $N = \{1, 2\}$, strategy set $S_i = \{2, \dots, 100\}$ for every player $i$, and payoff function $p_i$ for every $i$ defined as
\[
p_i(s) :=
        \begin{cases}
        s_{i}       & \mathrm{if}\  s_{i} = s_{-i} \\
        s_{i} + b & \mathrm{if}\  s_{i} < s_{-i} \\
        s_{-i} - b    & \mathrm{otherwise,}
        \end{cases}
\]
where $b > 1$ is the bonus.

\begin{proposition}
The selfishness level of the Traveler's Dilemma game is $\frac{b-1}{2}$.
\end{proposition}
\begin{proof}
The unique social optimum of this game is $s = (100,100)$, while
$(2,2)$ is its unique Nash equilibrium. If player $i$ deviates from $s$ to a strategy $s'_i \leq 99$, while the other player remains at $100$,
the respective payoffs become $s'_i +b$ and $s'_i -b$, so the social welfare becomes $2 s'_i$.
So $\af_i(s'_i, s) = (s'_i + b - 100)/(200 - 2 s'_i)$.
The maximum, $\frac{b-1}{2}$,
is reached when $s'_i = 99$. So the claim follows by Theorem~\ref{thm:characterization}~(ii).
\end{proof}

Intuitively, this means that as the bonus $b$ increases a larger share of the social welfare needs to be used to ensure cooperation.

\subsection{Tragedy of the Commons}\label{exa:tragedy}

Assume that each player $i \in N = \{1, \dots, n\}$ has the real interval $[0,1]$ as its set of strategies.
Each player's strategy is his chosen fraction of a common resource. Let (see \citeR[Exercise 63.1]{Osb05} and \citeR[pp.~6--7]{TV07}):
$$
p_i(s) := \max\Big(0,\ s_i \Big(1 - \sum_{j = 1}^{n} s_j\Big)\Big).
$$
 This payoff function reflects the fact that player's enjoyment of
 the common resource depends positively from his chosen fraction of the
 resource and negatively from the total fraction of the common
 resource used by all players. Additionally, if the total
 fraction of the common resource by all players exceeds a feasible
 level, here 1, then player's enjoyment of the resource becomes zero.

\begin{proposition}
\label{prop:tragedy}
The selfishness level of the $n$-players Tragedy of the Commons game is $\infty$.
\end{proposition}
Intuitively, this result means that in this game no matter how much we `involve' the players in 
sharing the social welfare we cannot achieve that they will select
a social optimum.

\begin{proof}
We first determine the stable social optima of this game. Fix a
joint strategy $s$ and let $t := \sum_{j \in N} s_j$.  
If $t > 1$, then the social welfare is $0$. So assume that $t
\leq 1$.  Then $SW(s) = t (1-t)$.  
This expression becomes maximal precisely when $t = \frac{1}{2}$ and
then it equals $\frac{1}{4}$.  So this game has infinitely many social
optima and each of them is stable.

Take now a stable social optimum $s$. So $\sum_{j \in N} s_j = \frac{1}{2}$.
Fix $i \in \{1, \dots, n\}$.
Denote $s_i$ by $a$ and consider a strategy $x$ of player $i$ such that $p_i(x, s_{-i}) > p_i(a, s_{-i})$.
Then $\sum_{j \neq i} s_j + x \neq \frac{1}{2}$, so $SW(a, s_{-i}) > SW(x, s_{-i})$.

We have $p_i(a, s_{-i}) = \frac{a}{2}$ and $SW(a, s_{-i}) = \frac{1}{4}$.
Further, $p_i(x, s_{-i}) > p_i(a, s_{-i})$ implies $\sum_{j \neq i} s_j + x < 1$ and hence
\[
p_i(x, s_{-i}) = x (a + \tfrac{1}{2} - x)
\quad\text{and}\quad
SW(x, s_{-i}) = (\tfrac{1}{2} - a + x) (1 - \tfrac{1}{2} + a - x) = \tfrac{1}{4} - (a-x)^2.
\]

Also $x \neq a$. Hence
\begin{align*}
\af_i(x,s) = & \frac{p_i(x, s_{-i}) - p_i(a, s_{-i})}{SW(a, s_{-i}) - SW(x, s_{-i})} = \frac{(a-x) (x - \frac{1}{2})}{(a-x)^2} = \frac{x - \frac{1}{2}}{a-x} = -1 + \frac{a - \frac{1}{2}}{a-x} 
\end{align*}

Since $p_i(x, s_{-i}) - p_i(a, s_{-i}) = (a-x)(x - \tfrac{1}{2})$ we have
$p_i(x, s_{-i}) > p_i(a, s_{-i})$ iff $a < x < \frac{1}{2}$ or $a > x > \frac{1}{2}$.
But $a \leq \frac{1}{2}$, since $\sum_{j \neq i} s_j + a =
\frac{1}{2}$.  So the conjunction of
$p_i(x, s_{-i}) > p_i(a, s_{-i})$ and $SW(x, s_{-i}) < SW(a, s_{-i})$ holds
iff $a < x < \frac{1}{2}$.
Now
$
\max_{a < x < \frac{1}{2}} \af_i(x,s) = \infty.
$
But $s$ was an arbitrary stable social optimum, so the claim follows by Theorem \ref{thm:characterization}~(i).
\end{proof}

\subsection{Cournot Competition}\label{exa:competition}

We consider Cournot competition for $n$ firms with a linear inverse demand function
and constant returns to scale (see, e.g., \citeR[pp.~174--175]{JR11}).
So we assume that each player $i \in N = \{1, \dots, n\}$ has a strategy set $S_i = \mathbb{R}_+$ and 
payoff function
$
p_i(s) := s_i (a - b \sum_{j \in N} s_j ) - c s_i
$
for some given $a, b, c$, where $a > c \geq 0$ and $b > 0$.

The price of the product is represented by the expression $a - b
\sum_{j \in N} s_j$ and the production cost corresponding to the
production level $s_i$ by $c s_i$.  In what follows we rewrite the
payoff function as $p_i(s) := s_i (d - b \sum_{j \in N} s_j)$, where
$d: = a-c$.
{Note that the payoffs can be negative, which was not the case in the
tragedy of the commons game. Still the proofs are very similar for
both games.}

\begin{proposition}
The selfishness level of the $n$-players Cournot competition game is $\infty$.
\end{proposition}
\begin{proof} 
We first determine the stable social optima of this game. Fix a joint
strategy $s$ and let $t := \sum_{j \in N} s_j$.  Then $SW(s) = t(d -
b t)$.  This expression becomes maximal precisely when $t =
\frac{d}{2b}$.  So this game has infinitely many social optima and
each of them is stable.

Take now a stable social optimum $s$. So $\sum_{j \in N} s_j = \frac{d}{2b}$.
Fix $i \in N$.
Let $u := \sum_{j \neq i} s_j$.  
For every strategy $z$ of player $i$
\[
p_i(z,s_{-i}) = - b z^2 + (d - bu) z \quad\text{and}\quad
SW(z,s_{-i}) = - b z^2 + (d - 2bu) z + u(d - bu).
\]

Denote now $s_i$ by $y$ and consider a strategy $x$ of player $i$ such that $p_i(x, s_{-i}) > p_i(y, s_{-i})$.
Then $u + x \neq \frac{d}{2b}$, so $SW(y, s_{-i}) > SW(x, s_{-i})$.

We have
\begin{align*}
p_i(x, s_{-i}) - p_i(y, s_{-i}) &  = - b (x^2 - y^2) + (d - bu) (x-y)  \\
& 
= - b (x-y) (x + y + u - \tfrac{d}{b}) 
= - b (x-y) (x - \tfrac{d}{2b}),
\end{align*}
where the last equality holds since $u - \frac{d}{b} = - (y + \frac{d}{2b})$ 
on the account of the equality $u + y = \frac{d}{2b}$.

Further, 
\begin{align*}
SW(y, s_{-i}) - SW(x, s_{-i}) & = b (x^2 - y^2) - (d -2bu) (x-y) \\
& 
= b (x -y) (x + y + 2u - \tfrac{d}{b}) =  b (x-y)^2,
\end{align*}
where the last equality holds since $2u - \frac{d}{b} =  - 2y$ on the account of the equality $u + y = \frac{d}{2b}$.

We have $x \neq y$. Hence
\[
\af_i(x, s) = \frac{p_i(x, s_{-i}) - p_i(y, s_{-i})}{SW(y, s_{-i}) - SW(x, s_{-i})} = - \frac{x- \frac{d}{2b}}{x-y} = 
- 1 + \frac{y - \frac{d}{2b}}{y - x}.
\]

Since $p_i(x, s_{-i}) - p_i(y, s_{-i}) =  b (y - x)(x - \tfrac{d}{2b})$ we have $p_i(x, s_{-i}) - p_i(y, s_{-i}) > 0$ iff $y < x < \frac{d}{2b}$ or $y >  x > \frac{d}{2b}$.
But $y \leq \frac{d}{2b}$, since $u + y = \frac{d}{2b}$.
So the conjunction of
$p_i(x, s_{-i}) > p_i(y, s_{-i})$ and $SW(x, s_{-i}) > SW(y, s_{-i})$ holds
iff $y < x < \frac{d}{2b}$.
Now
$
\sup_{y < x < \frac{d}{2b}} \af_i(x,s) = \infty.
$
But $s$ was an arbitrary stable social optimum, so the claim follows by Theorem \ref{thm:characterization}~(i).
\end{proof}

This proof shows that for every stable social optimum $s$, for every player there exist deviating strategies with an arbitrary high appeal factor. In fact, $\lim_{x \rightarrow y^+} \af_{i}(x,s) = \infty$, i.e., the appeal factor of the deviating strategy $x$ converges to $\infty$ when it converges from the right to the original strategy $y$ in $s$.

\subsection{Bertrand Competition}\label{exa:bcompetition}

Next, we consider Bertrand competition, a game concerned with a simultaneous selection of prices for the same product by two firms (see, e.g.,~\citeR[pp.~175--177]{JR11}).
The product is then sold by the firm that
  chose the lower price.  In the case of a tie the product is sold by
  both firms and the profits are split. We assume that each firm has 
identical marginal costs $c > 0$ and no fixed cost, and that
each strategy set $S_i$ equals $[c, \frac{a}{b})$,
where $c < \frac{a}{b}$.
The payoff function for player $i  \in \{1, 2\}$ is given by
\[
p_i(s_i, s_{3-i}) := 
\begin{cases}
(s_i - c)(a - b s_i) &  \mbox{if $c < s_i < s_{3-i}$} \\
\frac12 (s_i - c)(a - b s_i) &  \mbox{if $c < s_i = s_{3-i}$} \\
0 &  \mbox{otherwise.} 
\end{cases}
\]

\begin{proposition}
\label{prop:bertrand}
The selfishness level of the Bertrand competition game is $\infty$.
\end{proposition}

\begin{proof}
Let $d := \frac{a + bc}{2b}$.
If $SW(s) > 0$, then $SW(s) =  (s_0 - c)(a - b s_0)$, where $s_0 := \min(s_1, s_2)$. 
Note that $d \in (c, \frac{a}{b})$, since by the assumption $bc < a$.
Hence $s$ is a social optimum iff $\min(s_1, s_2) = d$.

If $s$ is a social optimum with $s_1 \neq s_2$, then player $i$ with the larger
$s_i$ can profitably deviate to $s_{3-i}$ (that equals $d$), 
while $(s_{3-i}, s_{3-i})$ remains a social optimum.
So the only stable social optimum is $(d,d)$. 

Fix  $i \in \{1, 2\}$. Note that if $s_i$ is slightly lower than $d$, then $p_i(s_i, d) > p_i(d, d)$.
Further,
\[
\lim_{s_i \to d^-} 
(p_i(s_i, d) - p_i(d,d)) = \tfrac12 (d-c)(a - bd),
\quad\text{while}\quad
\lim_{s_i \to d^-} (SW(d,d) - SW(s_i, d)) = 0
\]
and $SW(d,d) - SW(s_i, d) \neq 0$ for $s_i \neq d$.
Hence
\[
\sup_{c  < s_i < d} \frac{p_i(s_i, d) - p_i(d,d)}{SW(d,d) - SW(s_i, d)} = \infty.
\]
The claim now follows by Theorem \ref{thm:characterization}~(i).
\end{proof}

\section{Extensions and Future Research Directions}
\label{sec:conclusions}

We introduced the selfishness level of a game as a new measure of
discrepancy between the social welfare in a Nash equilibrium and in a
social optimum.
Our studies reveal that the selfishness level often provides deeper insights into the characteristics that influence the players' willingness to cooperate.
We conclude by mentioning some natural extensions and future research directions.

\subsection{Extensions}

The definition of the selfishness level naturally extends to other solution concepts and other forms of games.  

\subsubsection{Mixed Nash Equilibria}

For mixed Nash equilibria we can simply adapt our definitions by stipulating that a
strategic game $G$ is $\alpha$-selfish if a \emph{mixed} Nash
equilibrium of $G(\alpha)$ is a social optimum, where now we also
allow social optima in mixed strategies. The selfishness level of $G$
is then defined as before in \eqref{eq:sel-lev}.

For example, with this notion the selfishness level of the Matching
Pennies game (Example~\ref{exa:matching}) is $0$ since its unique
mixed Nash equilibrium, $(\frac12 H + \frac12 T,\frac12 H + \frac12
T)$, is also a social optimum.  The Matching Pennies game has no pure
Nash equilibrium.  In contrast, the game from
{Example~\ref{exa:bad}} does have a pure Nash equilibrium. When we
use mixed Nash equilibria its selfishness level also becomes $0$.  So
in both games the selfishness level changed from $\infty$, when pure
Nash equilibria are used, to $0$, when mixed Nash equilibria are used.

Further, a finite selfishness level of a finite game can decrease when
we use mixed Nash equilibria. As an example consider the following
`amalgamation' of the Matching Pennies (with payoffs increased by 2)
and Prisoner's Dilemma \gs{(with payoffs increased by 1)} games:

\begin{center}
\begin{game}{4}{4}
       & $H$    & $T$    & $C$     & $D$\\
$H$   &$3,1$   &$1,3$   &$0,0$   &$0,0$ \\
$T$   &$1,3$   &$3,1$   &$0,0$   &$0,0$ \\
$C$   &$0,0$   &$0,0$   &$2,2$   &$0,3$\\
$D$   &$0,0$   &$0,0$   &$3,0$   &$1,1$
\end{game}
\end{center}


This game has a unique stable social optimum, $(C,C)$, and a unique
pure Nash equilibrium, $(D,D)$.  It is easy to check using Theorem
\ref{thm:characterization}~(ii) that its selfishness level is 1.  On
the other hand, when we use mixed Nash equilibria then the selfishness
level becomes 0.  Indeed, $(\frac12 H + \frac12 T,\frac12 H + \frac12
T)$ is both a mixed Nash equilibrium and a social optimum in mixed
strategies.

\subsubsection{Extensive Games}

We can also consider  extensive games and subgame perfect equilibria.
As an example consider the six-period version of the centipede game (see, e.g., \citeR{Osb05}):

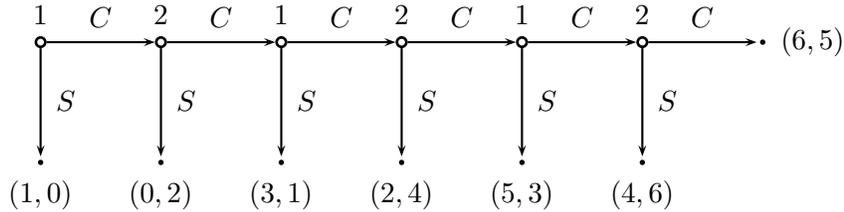
  \begin{center}
    \leavevmode
    \psset{unit=.8}
    \begin{pspicture}(.8,2.4)(12,5.6)
    \ownnodeL{.1}{white}{0,5}{u1}{$1$}{90}
    \ownnodeL{.1}{white}{2,5}{u2}{$2$}{90}
    \ownnodeL{.1}{white}{4,5}{u3}{$1$}{90} 
    \ownnodeL{.1}{white}{6,5}{u4}{$2$}{90} 
    \ownnodeL{.1}{white}{8,5}{u5}{$1$}{90}         
    \ownnodeL{.1}{white}{10,5}{u6}{$2$}{90}               
    \ownnodeL{0}{white}{12,5}{u7}{$(6,5)$}{0}                   
    
    \ownnodeL{0}{white}{0,3}{v1}{$(1,0)$}{-90}
    \ownnodeL{0}{white}{2,3}{v2}{$(0,2)$}{-90}
    \ownnodeL{0}{white}{4,3}{v3}{$(3,1)$}{-90} 
    \ownnodeL{0}{white}{6,3}{v4}{$(2,4)$}{-90} 
    \ownnodeL{0}{white}{8,3}{v5}{$(5,3)$}{-90}         
    \ownnodeL{0}{white}{10,3}{v6}{$(4,6)$}{-90}             
   
    \ncline[nodesep=.1]{->}{u1}{u2}\taput{$C$}
    \ncline[nodesep=.1]{->}{u2}{u3}\taput{$C$}
    \ncline[nodesep=.1]{->}{u3}{u4}\taput{$C$}        
    \ncline[nodesep=.1]{->}{u4}{u5}\taput{$C$}            
    \ncline[nodesep=.1]{->}{u5}{u6}\taput{$C$}         
    \ncline[nodesep=.1]{->}{u6}{u7}\taput{$C$}             
    
    \ncline[nodesep=.1]{->}{u1}{v1}\trput{$S$}
    \ncline[nodesep=.1]{->}{u2}{v2}\trput{$S$}
    \ncline[nodesep=.1]{->}{u3}{v3}\trput{$S$}        
    \ncline[nodesep=.1]{->}{u4}{v4}\trput{$S$}            
    \ncline[nodesep=.1]{->}{u5}{v5}\trput{$S$}                    
    \ncline[nodesep=.1]{->}{u6}{v6}\trput{$S$}                        
    \end{pspicture}
    \end{center}

In its unique subgame perfect equilibrium each player chooses $S$ in every period and the 
resulting payoffs are $(1,0)$. In contrast, the social optimum is obtained when each player
chooses $C$ in every period and the resulting payoffs are $(6,5)$. We seek 
$\alpha$ such that in the resulting game $G(\alpha)$ the latter 
pair of strategies forms a subgame perfect equilibrium. In particular, player 2 should choose in the
last round of $G(\alpha)$ the action $C$. This happens when 
$5 + (6+5)\alpha \geq 6 + (4+6)\alpha$ which holds iff $\alpha \geq 1$.
Now, for $\alpha = 1$ the game $G(\alpha)$ has the following payoffs:

  \begin{center}
    \leavevmode
    \psset{unit=.8}
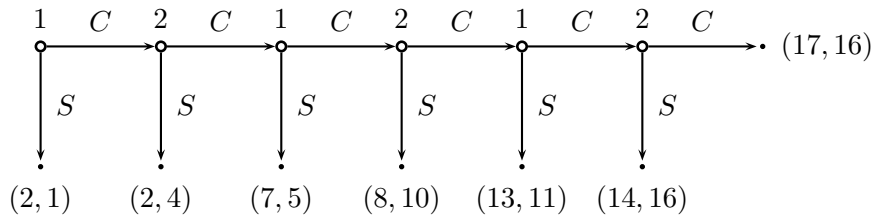
    \begin{pspicture}(.8,2.4)(12,5.6)
    \ownnodeL{.1}{white}{0,5}{u1}{$1$}{90}
    \ownnodeL{.1}{white}{2,5}{u2}{$2$}{90}
    \ownnodeL{.1}{white}{4,5}{u3}{$1$}{90} 
    \ownnodeL{.1}{white}{6,5}{u4}{$2$}{90} 
    \ownnodeL{.1}{white}{8,5}{u5}{$1$}{90}         
    \ownnodeL{.1}{white}{10,5}{u6}{$2$}{90}               
    \ownnodeL{0}{white}{12,5}{u7}{$(17,16)$}{0}                   
    
    \ownnodeL{0}{white}{0,3}{v1}{$(2,1)$}{-90}
    \ownnodeL{0}{white}{2,3}{v2}{$(2,4)$}{-90}
    \ownnodeL{0}{white}{4,3}{v3}{$(7,5)$}{-90} 
    \ownnodeL{0}{white}{6,3}{v4}{$(8,10)$}{-90} 
    \ownnodeL{0}{white}{8,3}{v5}{$(13,11)$}{-90}         
    \ownnodeL{0}{white}{10,3}{v6}{$(14,16)$}{-90}             
   
    \ncline[nodesep=.1]{->}{u1}{u2}\taput{$C$}
    \ncline[nodesep=.1]{->}{u2}{u3}\taput{$C$}
    \ncline[nodesep=.1]{->}{u3}{u4}\taput{$C$}        
    \ncline[nodesep=.1]{->}{u4}{u5}\taput{$C$}            
    \ncline[nodesep=.1]{->}{u5}{u6}\taput{$C$}         
    \ncline[nodesep=.1]{->}{u6}{u7}\taput{$C$}             
    
    \ncline[nodesep=.1]{->}{u1}{v1}\trput{$S$}
    \ncline[nodesep=.1]{->}{u2}{v2}\trput{$S$}
    \ncline[nodesep=.1]{->}{u3}{v3}\trput{$S$}        
    \ncline[nodesep=.1]{->}{u4}{v4}\trput{$S$}            
    \ncline[nodesep=.1]{->}{u5}{v5}\trput{$S$}                    
    \ncline[nodesep=.1]{->}{u6}{v6}\trput{$S$}                        
    \end{pspicture}
    \end{center}

So in this game the pair of strategies in which each player
chooses $C$ in every period is both a subgame perfect equilibrium and a 
social optimum and yields the payoffs $(17,16)$. 
We conclude that the (appropriately
adapted) selfishness level for this game is 1.

We leave for future work the study of such alternatives. 

\subsection{Future Research Directions}
\label{subsec:future}

\gs{There are several intriguing questions that we left open. We discuss a few future research directions below.}

\subsubsection{Abstract Games}

It would be interesting to define the notion of a selfishness level for
\bfe{abstract games}.  These are games in which the payoffs
are replaced by preference relations (see \citeR{OR94}).  By a \bfe{preference
  relation} on a set $A$ we mean here a linear ordering on $A$.  More
precisely, an abstract game is defined as 
$
(N, \{S_i\}_{i \in N}, \{ \succeq_i \}_{i \in N})
$ 
where each $\succeq_i$ is player's $i$
preference relation defined on the set $S_1
\times \dots \times S_n$ of joint strategies.
By a \bfe{realization} of an abstract game $(N, \{S_i\}_{i \in N}, \{ \succeq_i \}_{i \in N})$
we mean any strategic game $(N, \{S_i\}_{i \in N}, \{ p_i \}_{i \in N})$ such that for all $i \in N$
and $s, s' \in S_1 \times \dots \times S_n$ we have
$s \succeq_i s'$ iff $p_i(s) \succeq_i p_i(s')$.

Unfortunately, it is not clear how to do this. First, note that
the notion of a Nash equilibrium is well defined for abstract
games. However, there is no counterpart of the notion of a social
optimum, since there is no `global' preference relation on the set of
joint strategies.  

It is tempting to circumvent this difficulty by defining the notion of a
selfishness level of an abstract game $G$ using its realizations $G'$
and the corresponding games $G'(\alpha(G'))$, where $\alpha(G')$ is
the selfishness level of $G'$. Unfortunately the resulting 
strategic games $G'(\alpha(G'))$, where $G'$ is a realization of $G$ are not
realizations of a single abstract game, so this `detour' does not allow
us to associate with the initial abstract game another one.

As an example take two realizations of the abstract Prisoner's Dilemma 
game and the corresponding games $G'(\alpha(G'))$:
\begin{center}
\begin{game}{2}{2}
       & $C$    & $D$\\
$C$   &$2,2$   &$0,3$\\
$D$   &$3,0$   &$1,1$
\end{game}
\hspace*{1cm}
\begin{game}{2}{2}
       & $C$    & $D$\\
$C$   &$6,6$   &$3,6$\\
$D$   &$6,3$   &$3,3$
\end{game}
\end{center}

\begin{center}
\begin{game}{2}{2}
       & $C$    & $D$\\
$C$   &$2,2$   &$0,3$\\
$D$   &$2.5,0$   &$1,1$
\end{game}
\hspace*{1cm}
\begin{game}{2}{2}
       & $C$    & $D$\\
$C$   &$6,6$   &$3,6$\\
$D$   &$5,2.5$   &$3,3$
\end{game}
\end{center}

So both realizations have the selfishness level 1 but the transformed
games do not correspond to the same abstract game, since in the first
transformed game we have $p_2(D,C) \geq p_2(D,D)$, while in the second
one $p_2(D,D) > p_2(D,C)$.

\subsubsection{Selfishness Function}
In our approach we assigned to each game a positive real number, its selfishness level. A natural generalization of this idea would be to assign to each game $G$ the function $f_{G}: \mathbb{R}_+ \to \mathbb{R}_+$, where $f(\alpha)$ equals the price of stability of the game $G(\alpha)$. 
Then the selfishness level of $G$ is $\inf \{ \alpha \in \mathbb{R}_+ \; \mid \; f_{G}(\alpha) = 1\}$.

{The function $f_G$ has been studied for altruistic extensions of linear congestion games and fair cost sharing games \cite{CKKS11,EMAN10}. However, in these papers only upper bounds on $f_G$ are derived, which in light of the results obtained here cannot be tight. It would be interesting to determine $f_G$ exactly for these games. This would probably require a generalization of the characterization result presented in this paper.}

\subsubsection{Alternative Approach Based on the Price of Anarchy}

We defined the selfishness level of a game as the smallest $\alpha$
such that the price of stability of $G(\alpha)$ is $1$. Alternatively,
one might define the selfishness level as the smallest $\alpha$ such
that the price of anarchy of $G(\alpha)$ is $1$.  This alternative
approach often yields the value $\infty$.  Take for instance the
following coordination game $G$:

\begin{center}
\begin{game}{2}{2}
       & $A$    & $B$\\
$A$   &$1,1$   &$0,0$\\
$B$   &$0,0$   &$0,0$
\end{game}
\end{center}
Then for every $\alpha \geq 0$ \ $(A,A)$ is a social optimum in
$G(\alpha)$ with the social welfare $2 + 4\alpha$, while $(B,B)$ is a
Nash equilibrium in $G(\alpha)$ with the social welfare 0.  So this
alternative selfishness level of the game $G$ is $\infty$, while the original
selfishness level is of course 0.

As another example consider the game $G$ below
left and the corresponding game $G(\alpha)$ below right:

\begin{center}
\begin{game}{2}{2}
       & $A$    & $B$\\
$A$   &$1,1$   &$3,0$\\
$B$   &$0,3$   &$0,0$
\end{game}
\hspace*{1cm}
\begin{game}{2}{2}
       & $A$    & $B$\\
$A$   &$1 + 2\alpha, 1 +2\alpha$   &$3 + 3\alpha, 3 + 3\alpha$\\
$B$   &$3\alpha, 3 + 3\alpha$   &$0,0$
\end{game}
\end{center}
Its selfishness level is 1, since this is the smallest value $\alpha$
for which $(A,B)$ is a Nash equilibrium in $G(\alpha)$. On the other
hand, if we focus on the price of anarchy, then we need to choose the smallest
$\alpha$ such that $(A,A)$ is not a Nash equilibrium in $G(\alpha)$
while $(A,B)$ is. This is the case iff 
$3\alpha > 1 + 2\alpha$, i.e., when $\alpha > 1$. So this alternative
selfishness level of the game $G$ is $1^+$.

In view of these examples we find this alternative approach not very promising.
Still, it might be interesting to clarify for which games it yields
finite values.

\subsubsection{Alternative Approach Based on Approximate Nash Equilibria}

As mentioned in the related work section, an alternative approach to measure the stability of equilibria of a game is the following. Given a payoff-maximization game $G = (N, \{S_i\}_{i \in N}, \{p_i\}_{i \in N})$, we call $G$ \bfe{$\varepsilon$-stable} for some $\varepsilon \ge 0$ if there exists a social optimum $s$ that is also a \emph{$(1+\varepsilon)$-approximate} Nash equilibrium, i.e., for every player $i \in N$ and every $s'_i \in S_i$, $(1+\varepsilon) p_i(s) \ge p_i(s'_i, s_{-i})$.\footnote{\gs{F}or cost-minimization games, we require that $c_i(s) \le
(1+\varepsilon) c_i(s'_i, s_{-i})$.}
We define the \bfe{stability level} of $G$ as the infimum over all $\varepsilon \ge 0$ such that $G$ is $\varepsilon$-stable. 
Intuitively, a stability level of $\varepsilon$ means that if we alter the players' incentives by scaling their payoffs by a factor of $(1+\varepsilon)$ then a social optimum is realized as a Nash equilibrium. 

It would be interesting to study how the stability level of a game relates to its selfishness level. Using the above definitions, it is easy to see that when a game $G$ admits a social optimum $s$ such that for every player $i \in N$ and every $s'_i \in S_i$, $p_i(s) \ge SW(s) - SW(s'_i, s_{-i})$, then 
$G$ is $\alpha$-stable if $G$ is $\alpha$-selfish.\footnote{\gs{F}or cost-minimization games, this inequality reads $c_i(s'_i, s_{-i}) \ge SC(s'_i, s_{-i}) - SC(s)$.} Said differently, the stability level of $G$ is at most its selfishness level. Similarly, when the reverse inequality holds then $G$ is $\varepsilon$-selfish if $G$ is $\varepsilon$-stable.

In particular, the above observation can be applied to fair cost sharing games, where for every joint strategy $s$ it holds that for every $i \in N$ and every $s'_i \in S_i$, $c_i(s'_i, s_{-i}) \ge SC(s'_i, s_{-i}) - SC(s)$ (see also \eqref{eq:diff-SC} in Section~\ref{exa:cost-sharing}). We conclude that the stability level of a fair cost sharing game $G$ is at most its selfishness level. As a consequence, our bounds on the selfishness level derived in Section~\ref{exa:cost-sharing} extend to the stability level in this case. Further, it is not hard to verify that the stability level for singleton cost sharing games is at least $\max\{0, \frac{1}{2} c_{\max}/c_{\min} -1\}$ and for cost sharing games with integer costs is at least $\max \{0, \frac{1}{2} L c_{\max} - 1\}$ by considering the examples given in the proofs of Proposition~\ref{prop:cost-sharing} and Proposition~\ref{prop:cost-sharing-general}, respectively. Thus, for these games the stability level coincides with the selfishness level. 

However, it can be seen that these two notions do not always coincide. The public goods game is another example where it holds that there exists a social optimum $s$ such that for every player $i \in N$ and every $s'_i \in S_i$, $p_i(s) \ge SW(s) - SW(s'_i, s_{-i})$ (see proof of Proposition~\ref{prop:public-good}). Thus, the stability level of this game is at most the selfishness level. In fact, simple calculations show that the stability level is $\max\{0, (1-\frac{c}{n})/c\} \le \max\{0, (1-\frac{c}{n})/(c-1)\}$, where the latter is the selfishness level of the game.

We leave it for future work to further investigate the stability level and its relation to the selfishness level.

\subsubsection{Other Social Welfare Functions}

In this paper we exclusively concentrated on social welfare functions which are defined as the sum of the individual payoffs of the players. We leave it for future research to study the selfishness level of games for other social welfare functions, e.g., maximizing the minimum payoff of all players.

\section*{Acknowledgements}
We acknowledge initial discussions with Po-An Chen and final discussions with Valerio Capraro. We also thank anonymous reviewers of the preliminary version for their valuable comments. We are particularly grateful to all three reviewers of JAIR for the most helpful remarks.

\nocite{FKKV13,AMT08,JF09,FT09}
\bibliography{e}
\bibliographystyle{theapa}

\end{document}